\documentclass{article}
\usepackage{amssymb}
\usepackage{amsmath}
\usepackage{enumerate}
\usepackage[dvips]{color}
\usepackage{graphicx}
\usepackage{amsfonts}
\usepackage{amsthm}

\nonstopmode

\newcommand{\R }{\ensuremath{\mathbb R}}

\newcommand{\F} {\ensuremath{\mathcal{F}}}

\newcommand{\Ft }{\ensuremath{\mathcal{F}_{t-1}}}
\newcommand{\Ftt }{\ensuremath{\mathcal{F}_{t}}}
\newcommand{\Lt }{\ensuremath{\mathcal{L}^0(\mathcal{F}_{t-1};\mathbb{R}^d)}}
\newcommand{\Ltt }{\ensuremath{\mathcal{L}^0(\mathcal{F}_{t};\mathbb{R}^d)}}
\newcommand{\oS }{\ensuremath{\overline{S}}}
\newcommand{\uS }{\ensuremath{\underline{S}}}

\newtheorem{theorem}{Theorem}[section]
\newtheorem{proposition}[theorem]{Proposition}
\newtheorem{lemma}[theorem]{Lemma}

\newtheorem{assumption}[theorem]{Assumption}

\newtheorem{corollary}[theorem]{Corollary}

\theoremstyle{definition}
\newtheorem{definition}[theorem]{Definition}
\newtheorem{notation}[theorem]{Notation}

\theoremstyle{remark}

\newtheorem{remark}[theorem]{Remark}
\begin{document}

\title{Arbitrage and Hedging in model-independent markets\\ with frictions}

\author{Matteo Burzoni\thanks{Acknowledgements: Research supported by Universit\`a degli studi di Milano and the ETH Foundation.} \\
{\small ETH Z\"{u}rich, email: matteo.burzoni@math.ethz.ch} }

\maketitle

\begin{abstract}
We provide a Fundamental Theorem of Asset Pricing and a Superhedging Theorem for a model independent discrete time financial market with proportional transaction costs. We consider a probability-free version of the Robust No Arbitrage condition introduced by Schachermayer in \cite{S04} and show that this is equivalent to the existence of Consistent Price Systems. Moreover, we prove that the superhedging price for a claim g coincides with the frictionless superhedging price of g for a suitable process in the bid-ask spread.  
\end{abstract}




\section{Introduction}
The theory of Arbitrage and Hedging lies at the ground of any mathematical analysis of real world financial markets. It is therefore natural to consider the Fundamental Theorem of Asset Pricing (FTAP) and the Superhedging Theorem as the most important pillars of the modern Mathematical Finance. The central role of these two aspects was already clear to De Finetti in his work on \emph{coherence} and \emph{previsions} (see \cite{deF}). In the case of $\Omega$ being a finite set of events a version of the FTAP has been proved by Harrison and Pliska \cite{HP81} and non-trivial extensions to the case of a general $\Omega$ are possible by introducing a reference probability measure $P$, as in the celebrated work of Dalang-Morton-Willinger \cite{DMW90}. The superhedging Theorem can be subsequently obtained through the FTAP using classical arguments. Later on, frictions in the market have been included for a more realistic description. A first comprehensive study is due to Jouini and Kallal \cite{JK95} which generated a very rich offspring of papers. The setting proposed by Kabanov et al. (see e.g. \cite{KS01a,KRS02}), based on \textit{solvency cones}, allowed for the extension of the aforementioned classical results on the Fundamental Theorem of Asset Pricing with $\Omega$ finite, as in \cite{KS01a}, and with a general space $(\Omega,\mathcal{F},P)$, as in \cite{S04}. A great amount of literature is also available on the superreplication problem, for a non exhaustive list, see \cite{BT00, CK96, CPT99, LS97, K99, SSC95, S14}.\\
Nevertheless, the existence of a reference probability has been recently criticized and opened new and interesting challenges in several branches of Mathematical Finance under the name of Knightian Uncertainty. We conduct our study in this framework, in particular we do not fix, a priori, any class of probabilities.

\paragraph{Arbitrage and Consistent Price Systems} In this paper we consider a model-independent version of the Robust No Arbitrage condition introduced in \cite{S04}. Whenever this condition holds the broker still have room for proposing a discount on the bid-ask spread without creating with this operation arbitrage opportunities. In this sense the terminology ``robustness'' of the No Arbitrage condition should be interpreted rather than in the sense of a probability-free setup. The results of \cite{S04} connect the absence of arbitrage to the existence of a price process $S$ with values in the bid-ask spread which is a martingale under a certain risk-neutral probability $Q$. We call the couple $(Q,S)$ ``Consistent Price System'' (CPS) and it is said to be \emph{strictly consistent} if $S$ takes values in the relative interior of the bid-ask spread. Differently from the approach of \cite{S04} we are not defining arbitrage in terms of physical units of assets, but we are choosing a numer\`aire and evaluating a sure gain in terms of the value process of a certain strategy. Nevertheless we show in Section \ref{secFTAP} the analogous equivalence under the name of FTAP:
\begin{equation}\label{textFTAP}
\text{\emph{Robust No Model Independent Arbitrage iff there exists a strictly CPS}}.
\end{equation} 
Only a very short literature is available for these problems under Knightian uncertainty. When a class of (possibly non-dominated) set of priors $\mathcal{P}$ is considered, recent results in this direction are given by Bayraktar and Zhang \cite{BZ13} and Bouchard and Nutz \cite{BN14}. In \cite{BN14} a (non-dominated) version of No Arbitrage of the second kind ($NA_2(\mathcal{P})$) introduced in \cite{R09} is studied.\\
In \cite{BZ13} the authors considered the generalization to this framework of the concept of No Arbitrage ($NA(\mathcal{P})$) and No strict Arbitrage ($NA^s(\mathcal{P})$) used in \cite{KS01a,KRS02}. In a first version of the paper they considered a market with a single risky asset and by using a strong continuity assumption and a non-dominated version of the martingale selection problem, they were able to show a Fundamental Theorem of Asset Pricing. In the revised version they removed the continuity hypothesis and extended the previous result to the case of a multi-dimensional market. In this paper we are using a different notion of arbitrage and a probability-free setup, so that the two results are not directly comparable but, similarly to their approach, we are considering a modification of the bid-ask spread in order to individuate the set of CPSs. The goal is to explicitly construct an arbitrage opportunity, when the set of CPSs is empty, and this requires to tackle directly the dynamic multi-period problem. To this aim we will make use of the general theory of random sets which have already been considered by Rokhlin in \cite{Rk08} for the probabilistic case. Nevertheless in \cite{Rk08} the author provided an equivalent condition to the existence of CPSs based on random sets. This condition turns out to be also equivalent to  Robust No Arbitrage due to the equivalence \eqref{textFTAP} which was already known from \cite{S04}. Since in this paper we do not have \eqref{textFTAP} while, on the contrary, it is exactly what we want to show, the extension to the model-free setup of some results of \cite{Rk08} is only partially useful.\\
\paragraph{Super-hedging Theorem} The second part of this paper is devoted to the analysis of the Superhedging problem. Likewise the case of the Fundamental Theorem of Asset Pricing there are very few results in the model-free case. A first important paper on this topic is given by Dolinsky and Soner \cite{DS14} where the case of a discrete time single-asset market is considered with constant proportional transaction costs. By defining a Monge-Kanotorovich optimization problem and exploiting optimal transport techniques the authors succeeded to show that the superhedging price of a path-dependent European option $g$ coincides with the supremum of the expectations of $g$ in the set of proability measure called \textit{approximate martingale measures}. Roughly speaking a probability measure belongs to this set if for any $u\geq t$, the conditional expectation of $S_u$ at time $t$ is contained in the interval $((1-k)S_{t},(1+k)S_{t})$ where the constant $k$ models the proportional transaction costs. A very recent paper by Bartl, Cheridito, Kupper and Tangpi \cite{CKT15} consider some extensions of these results to the case of countably many trading dates and $d$ assets with constants $k_i$ for $i=1,\ldots,d$ modelling proportional transaction costs. In both cases a version of the Fundamental Theorem of Asset Pricing is derived from the superhedging duality. The continuous time case with a single risky asset is investigated in \cite{DS15}.\\
In this paper we consider the model-free hedging problem in a $d$-dimensional discrete time setting with random proportional transaction costs. The value process of a certain admissible strategy $H\in\mathcal{H}$ is evaluated in terms of units of a specified num\'eraire. In particular denoting by $\uS_t^j$ and $\oS^j_t$ the cost of selling and buying a share of asset $j$ at time $t$, we have that the value process $V_T(H)$ can be written as \begin{equation*} V_T(H)=\sum_{t=0}^T\sum_{j=1}^d\left(H^j_t-H^j_{t+1}\right)\left(\oS^j_t\mathbf{1}_{\{H^j_t\leq H^j_{t+1}\}}+\uS^j_t\mathbf{1}_{\{H^j_{t+1}\leq H^j_{t}\}}\right),
\end{equation*}
where we assume that $H_0=H_{T+1}=0$. Denote by $\mathcal{Q}$ the class of probability measures $Q$ that admits a consistent price system $(Q,S)$ for some $S$ with values in the bid-ask spread. In Section \ref{sec_super} we prove the following equality
\begin{equation}\label{textSuper}
\sup_{Q\in\mathcal{Q}}\mathbb{E}_Q[g]=\inf\{x\in\mathbb{R}\mid \exists H\in\mathcal{H}\textrm{ s.t. }x+V_T(H)\geq g\quad \forall\omega\in\Omega_*\}=:\overline{p}(g)
\end{equation}
for a measurable contingent claim $g$. The set $\Omega_*\subseteq \Omega$ for which we require the superhedging inequality is given by
$$\Omega_*:=\left\{\omega\in\Omega\mid \exists Q\in\mathcal{Q} \text{ such that }Q(\{\omega\})>0\right\}$$
and we denominate it the efficient support of the class of consistent price system CPS (See Definition \eqref{Omegastar}).
The reason for not considering the whole path space $\Omega$ but rather an \emph{efficient} subset of that is the existence, in the frictionless case, of examples which exhibit a duality gap for the analogous duality (see e.g. \cite{BFM15} and \cite{BNT16} in the context of martingale optimal transport).\\
The idea of the proof is the following. We first construct an auxiliary $\mathbb{S}$-superhedging problem (see Definition \ref{defF}) by considering at a certain time $t$ the whole set of random vectors which are convex combinations of random vectors at time $t+1$. Note that for such processes an obvious (conditional) martingale measure with finite support exists and it is specified by the convex combination. In Section \ref{sec_super} we show that by solving the $\mathbb{S}$-superhedging problem we obtain a process $S$ with values in the bid-ask spread and a trading strategy $H\in\mathcal{H}$ with the following property: denoting by $\overline{p}_{S}(g)$ the frictionless superhedging price for $g$, then the initial capital $\overline{p}_{S}(g)$ allows for superreplicating $g$ on $\Omega_*$ adopting the strategy $H$. By exploiting results from the frictionless case we will then get $\sup_{Q\in\mathcal{Q}}\mathbb{E}_Q[g]\geq \overline{p}(g)$ which is the difficult part in showing \eqref{textSuper}.\\
The rest of the paper is organized as follows: in Section \ref{settings} we introduce the framework, in Section \ref{secFTAP} we show the Fundamental Theorem of Asset Pricing and in Section \ref{sec_super} we study the Superhedging duality. Proofs of technical results from Section \ref{sec_super} are given in Section \ref{sec_proofs}. 

\section{Setting and notations}\label{settings}
Fix $(\Omega,\mathcal{B}(\Omega))$ a measurable space, where $\Omega$  is Polish, and $\mathcal{F}:=\mathcal{B}(\Omega)$ is the Borel sigma-algebra. Let $\mathfrak{P}=\mathfrak{P}(\Omega)$ be the set of probability measures on $(\Omega, \mathcal{F})$. We consider a discrete time interval $I=\{0,\ldots,T\}$ on a finite time horizon $T\in\mathbb{N}$ and we introduce a $(d+1)$-dimensional stochastic process $(\widetilde{S}_t)_{t\in I}$ which is Borel-measurable and which represents the discrete time evolution of the price process of $d+1$ assets where the first one serves as a numer\`aire. With no loss of generality we may therefore assume $\widetilde{S}^0_t\equiv 1$ for any $t\in I$. The setup of Kabanov et al. (for example \cite{KS01a,KRS02}) can be defined also when a reference probability is absent. For any $t\in I$, the cost for exchanging one unit of the asset $i$ for the corresponding value in units of the asset $j$, at time $t$, is specified by a Borel-measurable stochastic process $\lambda^{ij}_t$ for $i,j=0,\ldots,d$. Following the notation of Kabanov and Stricker \cite{KS01a} and Schachermayer \cite{S04}, one can also define the matrix $\Pi_t=[\pi^{ij}_t]_{i,j=0,\ldots,d}$ given by $$\pi_t^{ij}:=\dfrac{\widetilde{S}^j}{\widetilde{S}^i}(1+\lambda_t^{ij}),$$ where any $\pi^{ij}_t$ represents the physical unit of asset $i$ that an agent needs to exchange, at time $t$, for having one unit of asset $j$. Clearly $\lambda_t^{ii}=0$ and consequently $\pi_t^{ii}=1$ for any $t\in I$. A standard assumption is that agents are smart enough to take advantage of favourable exchange between assets so that, for any $t\in I$, for any $\omega\in\Omega$, one may assume $\pi_t^{ij}\leq \pi_t^{ik}\pi_t^{kj}$ for any $k=0,\ldots,d$.\\

In this paper the asset $\widetilde{S}^0$ serves as a numer\`aire and the value of any portfolio is calculated in terms of $\widetilde{S}^0$. This amounts to the choice of $\pi^{ij}_t=\pi^{i0}_t\pi^{0j}_t$ in the above setting for any $t\in I$ and $i\neq j$. We have therefore that the stochastic interval $[\frac{1}{\pi^{j0}},\pi^{0j}]$ represents the bid-ask spread of the asset $j\in\{1,\ldots,d\}$. 
 \begin{notation}In the following, the bid-ask spread  $\left[\frac{1}{\pi_t^{j0}},\pi_t^{0j}\right]$ will be shortly denoted as $[\uS^j_t,\oS^j_t]$ for $t=0,\ldots,T$ and $j=1,\ldots,d$ when it is more convenient.
\end{notation}

For any $t\in I$, for any $\omega\in\Omega$, define
\begin{equation}\label{convexBidAsk}
C_t(\omega):= \left[\uS_t^1(\omega),\oS_t^1(\omega)\right]\times,\ldots,\times\left[\uS_t^d(\omega),\oS_t^d(\omega)\right] \subseteq \mathbb{R}^d.
\end{equation}

\begin{assumption}\label{assumptions}
We model non-trivial transaction costs by assuming that $int(C_t)\neq \varnothing$ (efficient friction hypothesis), and we assume that, for every $\omega$ fixed, $C_t(\omega)$ is bounded.
\end{assumption}

We finally set $\mathbb{F}^{\widetilde{S}}:=\{\mathcal{F}^{\widetilde{S}}_t\}_{t\in I}$, where $\mathcal{F}^{\widetilde{S}}_t:=\sigma\{\uS_u,\oS_u\mid0\leq u\leq t\}$ denotes the natural filtration of the processes $\uS$ and $\oS$, and we consider the Universal Filtration $\mathbb{F}:=\{\mathcal{F}_{t}\}_{t\in I}$, namely,
\begin{equation*}
\mathcal{F}_{t}:=\bigcap_{P\in \mathfrak{P}}\mathcal{F}_{t}^{\widetilde{S}}\vee \mathcal{N%
}_{t}^{P},\text{ where }\mathcal{N}_{t}^{P}=\{N\subseteq A\in \mathcal{F}%
_{t}^{\widetilde{S}}\mid P(A)=0\}.
\end{equation*}%
Let $\mathbb{F}:=\{\Ftt\}_{t\in I}$. For any $0\leq t\leq T$, we denote by $\mathcal{L}^0(\mathcal{F}_{t};V)$ the set of \Ftt-measurable functions with values in $V\subseteq\R^d$. For technical purposes we will also adopt the following notation:
\begin{notation}\label{eps-dual}
For a random set $\Psi$ in $\R^d$ (see Definition \ref{defRandom} in the Appendix) we denote by $\Psi^*$ the (positive) dual of $\Psi$ and for $\varepsilon>0$ we introduce the $\varepsilon$-dual of $\Psi$ as \begin{eqnarray*}
\Psi^*(\omega)&:=&\{v\in\R^d\mid v\cdot x\geq 0 \quad \forall x\in\Psi(\omega) \},\\
\Psi^\varepsilon(\omega)&:=&\left\{v\in \mathbb{R}^d\mid v\cdot x \geq \varepsilon \quad \forall x\in\Psi(\omega)\setminus \{0\}\right\},
\end{eqnarray*} which they both preserve the same measurability of $\Psi$ as discussed in the Appendix (see Lemma \ref{meas-eps} and Proposition \ref{preservation_measurability}).
\end{notation}
\begin{notation}\label{notation}
 Throughout the text the following notations will be used: $co(\cdot)$, $conv(\cdot)$, $\overline{conv}(\cdot)$, $lin(\cdot)$, $ri(\cdot)$, which denote, respectively, the generated cone, the convex hull, the closure of the convex hull, the linear hull and the relative interior of a set. We use the notation $\R^{m\times n}$ for real matrices of $m$ rows and $n$ columns.
\end{notation}

\subsection{Arbitrage and Consistent Price Systems}\label{section_approach}
Differently from the frictionless case when an agent wants to implement a trading strategy she needs to consider the cost of rebalancing the portfolio after each trade date. The definition of self-financing strategies, goes as follows:
\begin{definition}\label{main_def} Denote by $e_i$ with $i=1,\ldots d$ the vector of the canonical base of $\mathbb{R}^{d}$ and define $$K_t:=co\left(conv\left\{e_i,\pi_t^{ij}e_i-e_j \mid i,j=1,\ldots,d\right\}\right),$$ the so-called \textbf{solvency cone}. Any portfolio in $K_t$ can be indeed reduced to the $0$ portfolio up to suitable exchanges of assets and up to ``throwing away'' some money if necessary. The \textbf{cone of portfolio available at cost 0} at time $t$, is simply given by $-K_t$ and  $F_t:=K_t\cap - K_t$ is the set of portfolio which are exchangeable with the zero portfolio.\\
A \textbf{self-financing} trading strategy $H:=(H_t)_{0\leq t\leq T+1}$ is an $\mathbb{F}$-predictable process with $H_0=H_{T+1}=0$ and $$H_t-H_{t-1}\in-K_{t-1}\quad \textrm{ for any }t=1,\ldots, T$$
meaning that rebalancing the portfolio is obtained at zero cost.
\end{definition}

We denote by $\mathcal{H}$ the class of self-financing strategies. Since $H_{T+1}=H_0+\sum_{t=0}^{T}\xi_t$ with $\xi_t\in -K_{t}$  any admissible strategy satisfies the following: i) it has no initial endowment ($H_0=0$); ii) at time $T$ any open position must be closed ($H_{T+1}=0$); iii) the portfolio is rebalanced, at zero cost, at any intermediate time.

We consider the value process $V_t(H)$ of a certain admissible strategy $H\in\mathcal{H}$ as the position in the num\'eraire $\widetilde{S}^0$ at time $t$ after rebalancing. The terminal value is given by
\begin{equation}\label{valueFunction} V_T(H)=\sum_{t=0}^T\sum_{j=1}^d\left(H^j_t-H^j_{t+1}\right)\left(\oS^j_t\mathbf{1}_{\{H^j_t\leq H^j_{t+1}\}}+\uS^j_t\mathbf{1}_{\{H^j_{t+1}\leq H^j_{t}\}}\right).
\end{equation}
One can easily verify the above formula. If, for instance, at time $t$ the agent switch from a long position to a short one in asset $j$ then she needs to liquidate $H^j_t$ obtaining $H^j_t\uS^j_t$ and then selling $H^j_{t+1}$ shares of the asset at the same price, yielding $(H^j_t-H^j_{t+1})\uS^j_t$ which coincides with the second term in \eqref{valueFunction} since obviously $H^j_{t+1}\leq H^j_t$. If instead she wants only to diminish the amount of shares in the long position, then $H^j_{t+1}\leq H^j_t$ and she needs to liquidate the amount $H^j_t-H^j_{t+1}$ obtaining in return $(H^j_t-H^j_{t+1})\uS^j_t$. The remaining cases follow similarly.\\

Using a similar argument as in Schachermayer \cite{S04} we may introduce, and motivate, the following definition of No Arbitrage,
\begin{definition}\label{smallerTr}
We say that a bid-ask process $\widetilde{\Pi}$ has smaller transaction costs than $\Pi$ if and only if for any $\omega\in\Omega$, for any $t\in I$ $$\left[\frac{1}{\widetilde{\pi}_t^{j0}},\widetilde{\pi}_t^{0j}\right]\subset\left(\frac{1}{\pi_t^{j0}},\pi_t^{0j}\right)
\quad \text{ for any }j=1,\ldots,d.$$
\end{definition}

Observe that clearly $V_T(H)$ depends also on $\Pi$ and, in particular, $V_T(H)(\widetilde{\Pi})> V_T(H)(\Pi)$ if $\widetilde{\Pi}$ has smaller transaction costs than $\Pi$ and $H$ is not the zero strategy. We will omit this dependence when it is clear from the context.

\begin{definition}
A Model Independent Arbitrage, with respect to a bid-ask spread $\widetilde{\Pi}$, is a trading strategy $H\in\mathcal{H}$ which satisfies $V_T(H)(\widetilde{\Pi})>0$ for any $\omega\in\Omega$.
\end{definition}
\begin{definition}\label{MIArb}
 Consider a market with bid-ask spread $\Pi$. We say that the market satisfies the Robust No Model Independent Arbitrage condition if there exists a bid-ask spread process $\widetilde{\Pi}$, with smaller transaction costs, for which there is No Model Independent Arbitrage with respect to $\widetilde{\Pi}$.
\end{definition}
This definition is a model-free version of the No Robust Arbitrage condition (NA$^r$) introduced in \cite{S04}. If the condition No Robust Arbitrage holds the broker still have room for proposing a discount on the transaction costs without creating arbitrage opportunities. On the contrary if this condition is not satisfied it is sufficient to have an infinitely small discount to get an arbitrage opportunity on a certain set of events. Since transaction costs are often subject of negotiation it looks quite natural to consider markets that exclude these possibilities.

\bigskip

We lastly need to formulate the definition of the so-called consistent price systems, in this model-free context.
\begin{definition}\label{defCPS} We say that a couple $(Q,S)$ is a consistent price system on $[0,T]$ if $S:=(S_t)_{t\in I}$ is a $(d+1)$-dimensional, $\mathbb{F}$-adapted stochastic process with $S_t^0\equiv 1$, for any $t\in I$ and which is a martingale under the measure $Q\in\mathfrak{P}(\Omega)$. In addition $S^j_t$ takes values in the bid ask-spread defined by $\Pi$, that is, $$S^j_t\in\left[\dfrac{1}{\pi_t^{j0}}\ ,\pi_t^{0j}\right],$$ for any $\omega\in\Omega$ and for any $j=1,\dots,d$.\\
The couple $(Q,S)$ is strictly consistent if $S$ takes values in the interior of the bid ask-spread.
Denote by $\mathcal{M}_{\overline{\Pi}}$ ($\mathcal{M}_{\Pi}$) the class of  price systems (strictly) consistent with $\Pi$.
\end{definition}

\begin{remark}We are considering the notion of generalized conditional expectation (see for example
\cite{FKV09}) where, for $\mathcal{G}\subseteq\mathcal{F}$, the expectation of a non negative $X\in \mathcal{L}^0(\mathcal{F};\mathbb{R})$ is defined by: $E_{P}[X\mid \mathcal{G}]:=\lim_{n\rightarrow +\infty
}E_{P}[X\wedge n\mid \mathcal{G}]$ and for $X\in \mathcal{L}^0(\mathcal{F};\mathbb{R})$ by $%
E_{P}[X\mid \mathcal{G}]:=E_{P}[X^{+}\mid
\mathcal{G}]-E_{P}[X^{-}\mid \mathcal{G}]$ with the
convention $\infty -\infty =-\infty $. All basic properties of the
conditional expectation still hold. In particular for a martingale measure $Q$ for $S$ and for a predictable $H$, we have $E_{Q}[H_{t}\cdot (S_{t}-S_{t-1})\mid \mathcal{F}%
_{t-1}]=H_{t}\cdot E_{Q}[(S_{t}-S_{t-1})\mid \mathcal{F}_{t-1}]=0$
$Q$-a.s.
\end{remark}

\section{Model free FTAP}\label{secFTAP}
We are now ready to introduce one of our main results.
\begin{theorem}[FTAP]\label{FTAP}Let $\mathcal{M}_{\Pi}$ the set of strictly consistent price systems as in Definition \ref{defCPS}.\begin{center} Robust No Model Independent Arbitrage holds $\Leftrightarrow$ $\mathcal{M}_{\Pi}\neq \varnothing$.\end{center}
\end{theorem}
\begin{proof}[Proof of ($\Leftarrow$)] Suppose $\mathcal{M}_{\Pi}\neq\varnothing$, hence there exist $S=(S_t)_{t\in I}$ and $Q\in\mathfrak{P}$ such that $S_t\in int(C_t)$ for $t\in I$, and $S$ is a $Q$-martingale. Consider $\widetilde{\Pi}$ a bid-ask process with smaller transaction costs for which, the corresponding $\widetilde{C}_t$ as in \eqref{convexBidAsk}, satisfies $S_t\in \widetilde{C}_t$ for $t\in I$. Let $H\in\mathcal{H}$ such that $V_T(H)\geq 0$. We note that
\begin{equation}\label{frictionleassDominance}
V_T(H)\leq (H\circ S)_T\ ,
\end{equation}
where $(H\circ S)_T$ is the usual (discrete time) stochastic integral. Equation \eqref{frictionleassDominance} is obtained by adding and subtracting $S^j_t$ in \eqref{valueFunction} and rearranging terms as follows (recall that $H_0=H_{T+1}=0$)
\begin{eqnarray*}
   V_T(H)&=&\sum_{t=0}^T\sum_{j=1}^d\left(H^j_t-H^j_{t+1}\right)\left(\oS^j_t\mathbf{1}_{\{H^j_t\leq H^j_{t+1}\}}+\uS^j_t\mathbf{1}_{\{H^j_{t+1}\leq H^j_{t}\}}-S^j_t+S^j_t\right)\\
   &=&\sum_{t=0}^T\sum_{j=1}^d\left(H^j_t-H^j_{t+1}\right)S^j_t+\\
   &&\sum_{t=0}^T\sum_{j=1}^d\left(H^j_t-H^j_{t+1}\right)\left((\oS^j_t-S^j_t)\mathbf{1}_{\{H^j_t\leq H^j_{t+1}\}}+(\uS^j_t-S^j_t)\mathbf{1}_{\{H^j_{t+1}\leq H^j_{t}\}}\right)\\
   &\leq&\sum_{t=0}^T\sum_{j=1}^d\left(H^j_t-H^j_{t+1}\right)S^j_t=\sum_{j=1}^d \left(\sum_{t=1}^T H^j_t S^j_t-\sum_{t=1}^T H^j_{t}S^j_{t-1}\right)=(H\circ S)_T.\\
  \end{eqnarray*}
 From $0\leq V_T(H)\leq (H\circ S)_T$, by taking expectations respect to $Q$, we get $V_T(H)=0$ $Q$-a.s. from which No Model Independent Arbitrage is possible.\\
\end{proof}
\begin{remark} In the frictionless case it has been shown in \cite{BFM14} that several concepts of arbitrage from the model-free context can be studied within the same framework, by means of the so-called \textit{Arbitrage de la classe} $\mathcal{S}$. We decided to choose the Model Independent notion, which is the strongest among this family (hence the weakest no arbitrage condition), and which correspond to $\mathcal{S}:=\{\Omega\}$. With similar techniques the analysis could be extended to general classes $\mathcal{S}$.
\end{remark}
Before giving the proof of the converse implication we need some preliminary results. This implication will be proven by contraposition, namely, assuming $\mathcal{M}_{\Pi}=\varnothing$ we will use an iterative modification of the bid-ask spread in order to capture arbitrage opportunities. This idea is similar in spirit to \cite{BZ13} but different in its implementation. In particular we do not solve first the problem for the one period case and then expanding to the multi-period case but we directly tackle the dynamic case. Note indeed that, when trading have transaction costs, arbitrage strategies might involve different times of execution. The simple example in the Introduction of \cite{BZ13} clarify this intuition: consider a single asset with deterministic bid-ask spread $[1,3]$ at time $0$ and $[2,4]$ $[3.5,5]$ at time $1$ and $2$ respectively. There is an arbitrage opportunity given by the strategy: buy at time 0 and sell at time 2.\\

For any $t\in I$, for any $\omega\in \Omega$, define iteratively, the following random sets
\begin{equation}\label{modifiedBidAsk}
\begin{split}
&\mathbb{S}_{T+1}(\omega):=\mathbb{R}^d\\
&\mathbb{S}_{t-1}(\omega):= C_{t-1}(\omega)\cap \overline{conv} \left(\mathbb{S}_{t}(\Sigma_{t-1}^{\omega}) \right)\quad \text{for }t=T+1\ldots,1
\end{split}
\end{equation}
 where $\Sigma_{t-1}^{\omega}=\{\widetilde{\omega}\in\Omega\mid \uS_{0:t-1}(\widetilde{\omega})=\uS_{0:t-1}(\omega), \oS_{0:t-1}(\widetilde{\omega})=\oS_{0:t-1}(\omega)\}\in\Ft$. Here $S_{0:t-1}(\omega)$ is a shorthand for the trajectory of the process $S$ up to time $t-1$.\\
 
The intuition behind this operation is the following. Consider first $t=T$ and observe that $\mathbb{S}_T$ is simply $C_T$. The random set $\mathbb{S}_{T-1}$ is given by the intersection of the bid-ask spread at time $T-1$ and the set of all convex combination of elements with values in the bid ask-spread at time $T$. Consider now a probability measure $P\in\mathfrak{P}$ with finite support and suppose $P(\Sigma_{T-1}^\omega)>0$. We note that if $P$ is a martingale measure for some $(S_{T-1},S_T)\in C_{T-1}\times C_{T}$ then $S_{T-1}$ needs to be a convex combination of $S_T$. We are therefore excluding from $C_{T-1}$ those values that cannot represents a conditional expectation of an $\mathcal{F}_{T}$-measurable random vector with values in $C_T$ respect to any probability measure with finite support. We first prove some measurability results.
 \begin{lemma}\label{LemmaTheta_meas}
  For any $t=0,\ldots, T+1$ the random set $\mathbb{S}_t$ as in \eqref{modifiedBidAsk} is \Ftt-measurable.
 \end{lemma}
 \begin{proof}
 For $t=T+1$ the claim is obvious. Suppose now that the claim holds for any $s\in\{t,\ldots, T+1\}$, we show that $\mathbb{S}_{t-1}$ is \Ft-measurable. Observe first that $C_{t-1}(\omega)$ is the closed convex hull of the multi-function $\omega\mapsto p_1(\omega) \times\cdots\times p_d(\omega)$ where $p_j=\{\uS_{t-1}^{j}\}\cup\{\oS_{t-1}^{j}\}$ for $j=1\ldots d$. All the $p_j$ are \Ft-measurable random sets being union of two \Ft-measurable random sets (whose values are singletons), by preservation of measurability through the operations of finite cartesian product, convex hull and closure we have that $C_{t-1}(\omega)$ is also \Ft-measurable (see Proposition \ref{preservation_measurability}).\\

 We turn now to the set $\mathbb{S}_t(\Sigma_{t-1}^{\omega})$. Denote by dom $\mathbb{S}_t:=\{\omega\mid \mathbb{S}_t(\omega)\neq\varnothing\}$ Since, by hypothesis, $\mathbb{S}_{t}$ is \Ftt-measurable it admits a Castaing representation, that is, there exists a collection $\{\varphi_n\}$ of \Ftt-measurable function $\varphi_n:$ dom $\mathbb{S}_t\rightarrow\mathbb{R}^d$ such that $\overline{\{\varphi_n(\omega)\mid n\in\mathbb{N}\}}=\mathbb{S}_t(\omega)$ for any $\omega\in\Omega$. Define therefore for $n\in\mathbb{N}$ the multi-functions $G_n:\omega\mapsto\{\varphi_n(\widetilde{\omega})\mid\widetilde{\omega}\in\Sigma_{t-1}^\omega\}$ which, as we now show, are \Ft-measurable: define $\gamma_{t-1}:\Omega\mapsto\R^{d\times t}\times\R^{d\times t}$ as $\gamma_{t-1}:=(\uS_{0:t-1},\oS_{0:t-1})$ and observe that $\forall O\subseteq\mathbb{R}^d$ open, and with $B:=\varphi_n^{-1}(O)$,  we have $$\{\omega\in\Omega\mid G_n(\omega)\cap O\neq \varnothing\}=\gamma_{t-1}^{-1}\left(\gamma_{t-1}(B)\right)\in\Ft$$
Recall indeed that image and counterimage of Borel sets through Borel measurable functions are analytic and that the Universal Filtration contains the class of analytic sets of \Ft (See for example Theorem III.18 and Theorem III.11 in \cite{DM82}). Observe now that $\overline{\mathbb{S}_t(\Sigma_{t-1}^\omega)}=\overline{\cup_{n\in\mathbb{N}}G_n}$. The inclusion $\supseteq$ is obvious. Take now $\overline{x}\in \overline{\mathbb{S}_t(\Sigma_{t-1}^\omega)}$ and a sequence $x_k\rightarrow \overline{x}$. We note that $x_k\in\overline{\cup_{n\in\mathbb{N}}G_n}$ for every $k$, since this set contains the collection $\{\varphi_n(\widetilde{\omega})\mid n\in\mathbb{N},\ \widetilde{\omega}\in\Sigma_{t-1}^{\omega}\}$ which is induced by the Castaing representation of $\mathbb{S}_t$. It therefore follows that $\overline{x}\in \overline{\cup_{n\in\mathbb{N}}G_n}$. We conclude that
  \begin{equation}\label{countableTheta}
  \mathbb{S}_{t-1}(\omega):= C_{t-1}(\omega)\cap \overline{conv} \left(\mathbb{S}_{t}(\Sigma_{t-1}^{\omega})\right)=C_{t-1}(\omega)\cap \overline{conv} \left(\cup_{n\in\mathbb{N}}G_n\right)
  \end{equation}
 is \Ft-measurable since the random sets $C_{t-1}$ and $\{G_n\}_{n\in\mathbb{N}}$ share the same measurability property and the transformations involved in \eqref{countableTheta} preserve measurability (see Proposition \ref{preservation_measurability}).
 \end{proof}
 \begin{corollary}\label{corTheta_Meas}
 The random sets
 $C_{t}(\omega)$, $\overline{\mathbb{S}_{t+1}(\Sigma_{t}^{\omega})}$ and $\overline{conv} \left(\mathbb{S}_{t+1}(\Sigma_{t}^{\omega}) \right)$ are \Ftt-measurable for any $t=0,\ldots, T$.
 \end{corollary}
 \begin{proof}
 Measurability of $C_t$ follows from the first part of the proof of Lemma \ref{LemmaTheta_meas}, measurability of $\overline{\mathbb{S}_{t+1}(\Sigma_{t}^{\omega})}$, and therefore of $\overline{conv} \left(\mathbb{S}_{t+1}(\Sigma_{t}^{\omega}) \right)$, follows from \eqref{countableTheta} and the discussion right before.
 \end{proof}

 \begin{remark}\label{smallerBidAsk} Note that with no loss of generality we may assume that if $\mathbb{S}_t(\omega)\neq\varnothing$ then $int(\mathbb{S}_t(\omega))\neq \varnothing$. For $t=T$ this is true since, by construction, $\mathbb{S}_T=C_T$ and $int(C_T)\neq\varnothing$ by Assumption \ref{assumptions}. If this is true up to time $t+1$ then it is true for time $t$ by considering, if needed, a bid-ask spread with smaller transaction costs $\widetilde{\Pi}$. Indeed, since $C_t$ and $\overline{conv} \left(\mathbb{S}_{t+1}(\Sigma_{t}^{\omega}) \right)$ have non empty interior by hypothesis, if the intersection has empty interior it is sufficient to consider an arbitrary small reduction of the bid-ask spread process to obtain $\mathbb{S}_t=\varnothing$. Take for example $\widetilde{\pi}_t^{0j}:=\pi_t^{0j}-\varepsilon^j(\omega)$ and $1/\widetilde{\pi}_t^{j0}:=1/\pi_t^{j0}+\varepsilon^j(\omega)$ where $\varepsilon^j(\omega):=\varepsilon \left(\pi_t^{0j}(\omega)-1/\pi_t^{j0}(\omega)\right)>0$ for an arbitrary small $\varepsilon>0$.
 \end{remark}
More formally we can consider $\widetilde{\Pi}$ with smaller transaction costs as in Definition \ref{smallerTr}, and define the corresponding $\widetilde{\mathbb{S}}_t$ as in \eqref{modifiedBidAsk} and $\widetilde{C}_t$ as in \eqref{convexBidAsk}, with bid-ask process $\widetilde{\Pi}$. Our aim is to show that under the assumption $\mathcal{M}_{\Pi}=\varnothing$ (with the original bid-ask process $\Pi$), there exists $H\in\mathcal{H}$ such that $V_T(H)(\widetilde{\Pi})>0$ for any $\omega\in\Omega$. Since we take $\widetilde{\Pi}$ arbitrary the thesis of the FTAP will follow. Observe first the following
\begin{lemma}\label{convexAndMtg}Let $\widetilde{\mathbb{S}}_t$ as in \eqref{modifiedBidAsk} with bid-ask process $\widetilde{\Pi}$ and $\mathcal{M}_{\Pi}$ the set of strictly consistent price systems as in Definition \ref{defCPS} for the bid-ask process $\Pi$. Then, \begin{equation*}
\{\omega\in\Omega\mid \widetilde{\mathbb{S}}_t(\omega)\neq\varnothing\quad \forall t=0\ldots T\}\neq\varnothing\Longrightarrow\mathcal{M}_{\Pi}\neq\varnothing.
\end{equation*}
\end{lemma}
\begin{proof} We build up a strictly consistent price system iteratively. Fix $\omega\in\Omega$ such that $\widetilde{\mathbb{S}}_t(\omega)\neq\varnothing$ for any $t=0,\ldots T$. By definition of $\widetilde{\mathbb{S}}_t$ we have $ri(\widetilde{\mathbb{S}}_t(\omega))=ri (\widetilde{C}_{t}(\omega))\cap ri (conv(\widetilde{\mathbb{S}}_{t+1}(\Sigma_{t}^{\omega}) ))$ (see e.g. Proposition 2.40 and 2.42 in \cite{R}) when the right hand side is non-empty. Note that, from $ri(\widetilde{C}_{t}(\omega))=int(\widetilde{C}_{t}(\omega))\subset int(C_t(\omega))$, we can assume this with no loss of generality. Indeed, if necessary, we can consider any $\hat{\Pi}$ which satisfies $\widetilde{C}_{t}(\omega)\subset int(\hat{C}_t(\omega))\subset int(C_t(\omega))$ for the corresponding set $\hat{C}_t(\omega)$. Therefore for any $\overline{y}\in ri(\widetilde{\mathbb{S}}_t(\omega))\neq\varnothing$ there exist $\lambda_1,\ldots, \lambda_m>0$ with $\sum_{i=1}^m\lambda_i=1$, $y_1,\ldots y_m\subseteq \R^d$, $\omega_1,\ldots \omega_m\subseteq\Sigma_{t}^\omega$ such that 
\begin{itemize}
\item $y_i\in \widetilde{\mathbb{S}}_{t+1}(\omega_i)\subseteq \widetilde{C}_{t+1}(\omega_i)\subset int(C_{t+1}(\omega_i))\quad \forall i=1,\ldots,m$
\item $\overline{y}=\sum_{i=1}^m\lambda_i y_i$
\end{itemize}
Start therefore with an arbitrary $x_0\in ri(\widetilde{\mathbb{S}}_0)\subset int(C_0)$ which is non-empty from the hypothesis. Associate to $x_0$ the real number $p({x_0})=1$ and set $Z_0=\{x_0\}$. Suppose a set of finite trajectories $Z_t:=\{x_{0:t}\in Mat(d \times (t+1))\}$ has been chosen up to time $t$ with associated $p(x_{0:t})>0$ summing up to one. By applying the above procedure to $x_t$ where $x_t$ is the value at time $t$ of a trajectory $x_{0:t}\in Z_t$, we can construct a new finite set of trajectories $$Z_{t+1}:=\{[x_{0:t},y_1(x_{0:t})],\ldots,[x_{0:t},y_m(x_{0:t})]\mid x_{0:t}\in Z_t\}$$ with associated $p([x_{0:t},y_i(x_{0:t})])=\lambda_i p(x_{0:t})$ for every $i=1,\ldots m(x_{0:t})$. \\ Observe that given the set $Z_T$ for any $x_{0:T}\in Z_T$ there exists $\omega\in\Omega$ such that $x_{0:T}\in C_0(\omega)\times\cdots\times C_T(\omega)$. Moreover, defining $S_t(\omega):=x_t\mathbf{1}_{\Sigma_{t}^{\omega}}$ and the probability measure $Q(\omega):=p(x_{0:T})$ we have that $S_t$ is \Ftt-measurable for any $t=0,\ldots T$ and \begin{equation}E_Q[S_{t}\mid\Ftt]=S_{t-1},\quad \text{ for t=1,\ldots T}. \end{equation} Thus, $Q$ is a martingale measure for $S$ which by construction lies in the interior of the bid-ask spread $\Pi$.
\end{proof}

\bigskip

 We are now able to complete the proof of Theorem \ref{FTAP}.\\

\begin{proof}[Proof of Theorem \ref{FTAP} $(\Rightarrow)$]
We prove the ``only if'' part by contraposition through several steps. Assume $\mathcal{M}_{\Pi}=\varnothing$ and let $\widetilde{\Pi}$ a bid-ask spread smaller than $\Pi$ with $\widetilde{C}_T\neq\varnothing$.\\
\textbf{Step 1:} Define first the random time $$\tau(\omega):=\inf\{0\leq t\leq T\mid \widetilde{\mathbb{S}}_{t}(\omega)=\varnothing\text{ and } \overline{conv} \left(\widetilde{\mathbb{S}}_{t+1}(\Sigma_{t}^{\omega}) \right)\neq\varnothing\}.$$
Observe that $\tau$ is a stopping time: for any $t\in I$ the set $\{\tau\leq t\}$ coincides with the set $\cup_{u=1}^t(\{\omega:\widetilde{\mathbb{S}}_u(\omega)=\varnothing\}\cap \{\omega:\overline{conv} (\widetilde{\mathbb{S}}_{u+1}(\Sigma_{u}^{\omega}) )\neq\varnothing\}$ which belongs to \Ftt$\ $from Lemma \ref{LemmaTheta_meas} and Corollary \ref{corTheta_Meas}. Observe now that under the assumption $\mathcal{M}_{\Pi}=\varnothing$, as a consequence of Lemma \ref{convexAndMtg}, for any $\omega$ there exists $u=u(\omega)$ such that $\widetilde{\mathbb{S}}_u(\omega)=\varnothing$. \\ Straightforward from definition \eqref{modifiedBidAsk}, $\widetilde{\mathbb{S}}_T(\omega)=\widetilde{C}_T(\omega)\neq\varnothing$ and hence $\overline{conv} (\widetilde{\mathbb{S}}_{T}(\Sigma_{T-1}^{\omega}))\neq \varnothing$. We can therefore deduce that $\tau(\omega)\leq T-1$ for any $\omega\in\Omega$, thus, $\tau$ is a finite stopping time.\\

\textit{Since for the rest of the proof we are considering the smaller bid-ask process $\widetilde{\Pi}$ for ease of notation and exposition we omit the superscript $\ \widetilde{\cdot}\ $ as no confusion arise here. So that we denote $\widetilde{\mathbb{S}}_t$ simply as $\mathbb{S}_t$ and $\widetilde{C}_t$ simply as $C_t$ for every $t\in I$.}\\
 
\textbf{Step 2:} For any $t\in\{0,\ldots T\}$ let $H=\{H_u\mid u\leq t\}$ with $H_0=0$ and $H_u\in \mathcal{L}^0(\F_{u-1};\R^d)$ be given. For $\xi:=sign H_t$, we introduce the following process $\hat{S}_t^{\xi}$ which take values at the boundary of the bid-ask spread.\footnote{The choice for the event ${\{H_t^i= 0\}}$ can be actually arbitrary without affecting the value of the strategy, for simplicity it is included here in the positive case.}
\begin{equation}\label{S_mod}
\hat{S}^{\xi}_t:=\left(\uS^1_t\mathbf{1}_{\{H_t^1\geq 0\}}+\oS^1_t
\mathbf{1}_{\{H_t^1<0\}},\ldots,\uS^d_t\mathbf{1}_{\{H_t^d\geq 0\}}+\oS^d_t
\mathbf{1}_{\{H_t^d<0\}}\right).
\end{equation}

We introduce also the sets $A_t$ and $B_t$ as follows: 
\begin{equation}\label{AandB}
A_t:=\{\tau=t\}\cap \bigcap_{u=0}^{t}\{H_u=0\},\qquad B_t:=\{H_t\neq 0\}\cap \{\hat{S}^{\xi}_t\notin \mathbb{S}_t\}.
\end{equation}
For an interpretation of these sets see Remark \ref{interpretationAB}.\\

We now show that $A_t$ and $B_t$ are \Ftt-measurable. The measurability of $A_t$ is obvious from $\tau$ being a stopping time and the measurability of $H_u$ for $u\leq t$. Now, observe that $sign(H_t)$ is \Ft-measurable since for any $x\in\Xi:=\{x\in\R^d\mid x^i\in\{-1,0,1\}\}$, $sign(H_t)^{-1}(x)=H_{t}^{-1}(x_1(0,\infty)\times,\ldots\times x_d(0,\infty))$ where with a slight abuse of notation $x_i(0,\infty)$ is either $(0,\infty)$, $(-\infty,0)$ or $\{0\}$ according to $x_i$ being respectively $1$, $-1$ or $0$.\\ $\hat{S}^{\xi}_t$ is \Ftt-measurable since for any $O:=O_1\times\ldots\times O_d\subseteq\mathbb{R}^d$ with $O_i$ open for $i=1,\ldots d$, we have
\begin{eqnarray*}
(\hat{S}^{\xi}_t)^{-1}(O)= &\bigcap_{i=1}^d\{&(\uS^i)^{-1}(O_i)\cap (\xi)^{-1}[0,\infty)\quad \cup \\& & (\oS^{i})^{-1}(O_i)\cap (\xi)^{-1}(-\infty,0)\qquad\}.
\end{eqnarray*}
The set $\{\hat{S}^{\xi}_t\in \mathbb{S}_t\}$ is \Ftt-measurable since it is the projection on $\Omega$ of the intersection of $Graph(\hat{S}^{\xi}_t)$ and $Graph(\mathbb{S}_t)$. We easily conclude that $B_t$ is \Ftt-measurable.\\

\textbf{Step 3:} Consider the sets $\{A_t\}_{t\in I}$ as in Step 2. We show that for any $t=1,\ldots T$ and for any $\varepsilon>0$, there exists an \Ft-measurable random vector $H^A_{t}$ such that $\forall\omega\in A_{t-1}$ \begin{equation}\label{ArbAt}
H^A_{t}(\omega)\cdot (s-x)\geq\varepsilon\qquad \forall s\in\mathbb{S}_{t}(\Sigma_{t-1}^\omega),\ \forall x\in C_{t-1}(\omega).
\end{equation}
To see this observe that the random set $(\overline{\mathbb{S}_{t}(\Sigma_{t-1}^\omega)}-C_{t-1}(\omega))^\varepsilon$ (see Notation \ref{eps-dual}) is closed-valued and \Ft-measurable by Corollary \ref{corTheta_Meas} and Lemma \ref{meas-eps}. It remains to show that it is non-empty for every $\omega\in A_{t-1}$ so that the desired $H^A_{t}$ is any measurable selector of this set. 
For any $\omega\in A_{t-1}$ we have $\mathbb{S}_{t-1}(\omega)=\varnothing$ and therefore, by \eqref{modifiedBidAsk}, the random sets $C_{t-1}(\omega)$ and $\overline{conv} \left(\mathbb{S}_{t}(\Sigma_{t-1}^{\omega}) \right)$ are closed, convex and disjoint. Hahn-Banach Theorem applies and for every $\omega\in A_{t-1}$ there exist $\varphi\in\R^d$, $l\in\R$ such that, in particular, $\inf\{\varphi\cdot s\mid s\in\overline{\mathbb{S}_{t}(\Sigma_{t-1}^{\omega})}\}>l> \sup\{\varphi\cdot x \mid x\in C_{t-1}(\omega)\}$. For a suitable $\gamma>0$, we also have $\varphi\cdot s>l+\gamma$, $\forall s\in\overline{\mathbb{S}_{t}(\Sigma_{t-1}^{\omega})}$. From $l>\varphi\cdot x$, $\forall x\in C_{t-1}(\omega)$, we obtain $\varphi\cdot (s-x)>\gamma$ for any $s\in\overline{\mathbb{S}_{t}(\Sigma_{t-1}^{\omega})},\ x\in C_{t-1}(\omega)$. Moreover, for any $\alpha>0$, $\alpha\varphi$ satisfies the same inequality with lower bound $\alpha\gamma$. We thus have the thesis with $\alpha=\varepsilon/\gamma$.\\
Let us stress that the value $\varepsilon$ in \eqref{ArbAt} can be arbitrary.
\bigskip

\textbf{Step 4:} We are now ready to construct iteratively an arbitrage opportunity which will satisfy, for an arbitrary $\delta>0$, the following:
 \begin{equation}\label{ArbValueInduction}
V_{t-1}(H)+H_{t}\cdot s\geq \frac{\delta}{2^{t-1}}\ \text{ for any }s\in \mathbb{S}_{t}(\Sigma_{t-1}^\omega)\text{ and for any }\omega\in A_{t-1}\cup B_{t-1}
\end{equation}
with $V_{t-1}(H)\leq 0$. For $t=1$ equation \eqref{ArbValueInduction} is trivially satisfied by $H_1:=H^A_1$ as in \eqref{ArbAt} with $\varepsilon=\delta$ arbitrary: we have indeed that $B_0=\varnothing$ and from \eqref{valueFunction} we can rewrite $V_0(H_1)+H_1\cdot s$ as $H_1\cdot(s-\hat{x})$ with
$\hat{x}:=\oS^j_0\mathbf{1}_{\{0\leq H^j_{1}\}}+\uS^j_0\mathbf{1}_{\{H^j_{1}\leq 0\}}\in C_0(\omega)$. From \eqref{ArbAt} the thesis follows. 

\bigskip

Suppose now we are given a strategy $H=(H_u)_{u=1}^{t}$ satisfying \eqref{ArbValueInduction}. \\ For any $\eta\in\Xi=\{x\in\R^d\mid x^i\in\{-1,0,1\}\}$ denote the partial order relation on $\R^d$ given by
$$h_1\preceq_\eta h_2\quad\text{ iff }\quad h_1-h_2\in \eta^1[0,\infty)\times\cdots\times \eta^d [0,\infty),$$
with the same slight abuse of notation of Step 2.\\ Similarly as in \eqref{S_mod} define $\hat{S}^{\eta}_t:=[\uS^j_t\mathbf{1}_{\{\eta^j\geq 0\}}+\oS^j_t
\mathbf{1}_{\{\eta^j<0\}}]_{j=1}^d$ and consider
\begin{equation}\label{deffeta}
f^\eta:=\omega\mapsto\left\{h\in\R^d \mid H_t(\omega)\preceq_\eta h \text{ and } V_{t}^h(H)+h\cdot s\geq \frac{\delta}{2^t}\quad \forall s \in \mathbb{S}_{t+1}(\Sigma_{t}^{\omega})\right\},
\end{equation}
where $V_{t}^h(H):=V_{t-1}(H)+(H_t-h)\cdot\hat{S}^\eta_t(\omega)$ is the value of the strategy $H=H_1,\ldots, H_t$ extended with $H_{t+1}(\omega)=h$ (cfr equation \eqref{valueFunction}). We here show that we can choose a measurable selector $H_{t+1}$ of $\cup_{\eta\in\Xi}f^\eta$ which we extend as $H_{t+1}=0$ on $\{\cup_{\eta\in\Xi}f^\eta=\varnothing\}$. In Lemma \ref{lemmaNonEmpty} we show that for any $\omega\in A_t\cup B_t$ such that $V_{t}(H)\leq 0$, for at least one $\eta\in\Xi$, the set $f^\eta$ is non-empty so that $(H_u)_{u=1}^{t+1}$ satisfy the desired inequality \eqref{ArbValueInduction} for time $t$. When $V_{t}(H)>0$ and $H_{t+1}=0$ the position is closed with a strictly positive gain.\\

Regarding measurability we consider the $(\delta/2^t)$-dual of the \Ftt-measurable random set $[\overline{\mathbb{S}_{t+1}(\Sigma_t^\omega)}-\hat{S}^\eta_t(\omega);V_{t-1}(H) + H_t \cdot\hat{S}^\eta_t]$ (see Corollary \ref{corTheta_Meas} and recall Notation \ref{eps-dual}), that is,
$$\left\{(h,h_{d+1})\in\R^d\times \R\mid h\cdot (s-\hat{S}^\eta_t(\omega))+h_{d+1}(V_{t-1}(H) + H_t \cdot\hat{S}^\eta_t)\geq \frac{\delta}{2^t}\quad \forall s \in \mathbb{S}_{t+1}(\Sigma_{t}^{\omega})\right\}$$ and we take the intersection with the closed-valued, \Ftt-measurable, random set $$\eta_1(-\infty,H^1_t(\omega)]\times,\ldots\times\eta_d(-\infty,H^d_t(\omega)]\times\{1\}.$$
By Proposition \ref{preservation_measurability} the finite union over $\eta\in\Xi$ is again closed-valued and \Ftt-measurable so that we can extract a measurable selection $\varphi$. A measurable selector of $\cup_{\eta\in\Xi}f^\eta$ is therefore given by the projection on the first $d$ components of $\varphi$.\\
 

\textbf{Step 5:} Let $H:=(H_u)_{u=1}^T$ the iterative strategy constructed in Step 4. For every $\omega\in\Omega$ we have $\tau(\omega)\leq T-1$ and $H_{\tau(\omega)+1}\neq 0$, that is, the position is opened at time $\tau$. Observe that if there exists $t\geq\tau(\omega)+1$ such that $\hat{S}^{\xi}_{t}$ defined in \eqref{S_mod} satisfies $\hat{S}^{\xi}_{t}(\omega)\in \mathbb{S}_{t}(\omega)$, then the position can be closed with a strictly positive gain. Indeed with $h=0$ we get, from \eqref{valueFunction} and from \eqref{ArbValueInduction},
\begin{equation}\label{liquidation}
V^h_t(H)=V_{t-1}(H)+\sum_{j=1}^d\left(H^j_t-0\right)\left(\oS^j_t\mathbf{1}_{\{H^j_t\leq 0\}}+\uS^j_t\mathbf{1}_{\{0\leq H^j_{t}\}}\right)
\geq \frac{\delta}{2^{t-1}}.
\end{equation}Note that from \eqref{AandB}, $H_u(\omega)=0$ for all $u\geq t+1$. Moreover, since $\mathbb{S}_T=C_T$ we obviously have $t\leq T$. Thus, the position given by strategy $H$ from Step 4, opened at time $\tau$, can always be closed with $V_{T}(H)(\omega)>0$. Since $\omega\in\Omega$ is arbitrary we have the conclusion.

\end{proof}

\begin{lemma}\label{lemmaNonEmpty}
Let $t\in I$ and $A_t$, $B_t$ from \eqref{AandB}. For any $\omega\in A_t\cup B_t$ fixed, the set $\cup_{\eta\in\Xi}f^\eta(\omega)$ is non-empty, where $f^\eta(\omega)$ is defined in \eqref{deffeta}.
\end{lemma}
\begin{proof} For $\omega\in A_t$ consider $H^A_{t+1}$ as in \eqref{ArbAt} with $\varepsilon=\delta/2^t$. The conclusion follows from $V_{t-1}(H_t)+H_t\cdot s=H_t\cdot(s-\hat{x})$ for 
$\hat{x}:=\oS^j_{t-1}\mathbf{1}_{\{0\leq H^j_{t}\}}+\uS^j_{t-1}\mathbf{1}_{\{H^j_{t}\leq 0\}}\in C_{t-1}(\omega)$. (cfr equation \eqref{valueFunction}).\\

We now turn to $\omega\in B_t$. Since $\hat{S}^{\xi}_t(\omega)\notin\mathbb{S}_t(\omega)$ the position cannot be closed without a loss at time $t$. We show that nevertheless it is possible to rebalance the portfolio in order to maintain a positive wealth. Consider set of vertices of $C_t(\omega)$
$$V:=\bigcup\left\{[\uS^j_t\mathbf{1}_{\{\eta^j\geq 0\}}+\oS^j_t
\mathbf{1}_{\{\eta^j<0\}}]_{j=1}^d\mid \eta\in\{-1,0,1\}^d\right\}$$
and the set $$L:=\{y\in\R^d\mid V_{t-1}(H)+H_t\cdot y\leq 0\}\cap V.$$
From the inductive hypothesis we have: i) $B_t\subseteq A_{t-1}\cup B_{t-1}$ since $H_t(\omega)\neq0$ only on $A_{t-1}\cup B_{t-1}$ and ii) $\mathbb{S}_t(\omega)\cap L(\omega)=\varnothing$. Moreover, since $\hat{S}^{\xi}_t(\omega)$ as in \eqref{S_mod} is a vertex and $V_t(H)\leq 0$, we thus have  $\hat{S}^{\xi}_t(\omega)\in L(\omega)$. Consider now the set
$$F:=\left\{h\in\R^d\setminus\{0\}\mid h\cdot(s-y)\geq 0\quad \forall s \in \mathbb{S}_{t+1}(\Sigma_{t}^{\omega}),\ \forall y\in L(\omega)\right\},$$ which is non-empty for $\omega\in B_t$: since $L(\omega)\subseteq C_t(\omega)$ and $L(\omega)\cap\mathbb{S}_t(\omega)=\varnothing$ then by \eqref{modifiedBidAsk} the sets $\overline{conv}(L(\omega))$ and $\overline{conv}(\mathbb{S}_{t+1}(\Sigma_{t}^{\omega}))$ are disjoint. Applying Hyperplane separating Theorem we obtain the assertion. Note, moreover, that since the separation is strict for any $h\in F$ there exists $\varepsilon>0$ such that $h\cdot(s-y)\geq \varepsilon\quad \forall s \in \mathbb{S}_{t+1}(\Sigma_{t}^{\omega}),\ \forall y\in L(\omega)$.\\

For any $h\in\R^d$ define now \begin{equation}\label{extension}
[\hat{S}^h_t]^j:=\oS^j_t\mathbf{1}_{\{H^j_t\leq h^j\}}+\uS^j_t\mathbf{1}_{\{h^j\leq H^j_{t}\}},
\end{equation}where $[\cdot]^j$ denotes the $j^{th}$ component of a vector.
We can distinguish two cases:
\begin{enumerate}
 \item there exists $h\in F$ such that $\hat{S}^h_t\in L$;\label{item1dimFTAP}
 \item for all $h\in F$, $\hat{S}^h_t\in V\setminus L$. \label{item2dimFTAP}
\end{enumerate}
  In case \ref{item1dimFTAP}. there exists $h\in F$ and $\varepsilon>0$ such that $h\cdot(s-\hat{S}^h_t)\geq \varepsilon$ for all $s \in \mathbb{S}_{t+1}(\Sigma_{t}^{\omega})$. Define now \begin{equation}\label{coeffMeas}
  \alpha_1:=\max\left\{\frac{1}{\varepsilon}\left(-V_{t-1}(H)-H_t\cdot \hat{S}^h_t+\frac{\delta}{2^{t}}\right),1+\frac{\delta}{2^{t}}\right\}\geq 1+\frac{\delta}{2^{t}}, \qquad \bar{h}:=\alpha_1 h\in F
  \end{equation} and observe that 
\begin{equation}\label{valueextension}
V_{t-1}(H)+H_t\cdot \hat{S}^h_t+\bar{h}\cdot(s-\hat{S}^h_t)\geq \frac{\delta}{2^{t}}\qquad \forall s \in \mathbb{S}_{t+1}(\Sigma_{t}^{\omega}).
\end{equation}
In order to retrieve the value $V^{\bar{h}}_t(H)$ in \eqref{valueextension} we need to replace $\hat{S}^h_t$ with $\hat{S}^{\bar{h}}_t$. By showing that $(H_t-\bar{h})\cdot \hat{S}^{\bar{h}}_t\geq (H_t-\bar{h})\cdot \hat{S}^h_t$, it will follow from \eqref{valueextension} that
\begin{eqnarray*}
V_{t}^{\bar{h}}(H)+\bar{h}\cdot s=V_{t-1}(H)+(H_t-\bar{h})\cdot \hat{S}^{\bar{h}}_t+\bar{h}\cdot s&\geq&\\ V_{t-1}(H)+(H_t-\bar{h})\cdot \hat{S}^{h}_t+\bar{h}\cdot s&\geq &\frac{\delta}{2^{t}} \qquad \forall s \in \mathbb{S}_{t+1}(\Sigma_{t}^{\omega})
\end{eqnarray*}
and hence the desired inequality. To show the claim let $j\in\{1,\ldots,d\}$. If $h^jH^j_t\leq 0$ or $|h^j|\geq |H^j_t|$ then from \eqref{extension} and $\alpha_1>1$ we get $[\hat{S}^{\bar{h}}_t]^j=[\hat{S}^h_t]^j$. Suppose now $H^j_t\leq h^j<0$ then again from \eqref{extension} and $\alpha_1>1$ we obtain $[\hat{S}^{\bar{h}}_t]^j\leq [\hat{S}^h_t]^j$ from which
$$(H^j_t-\bar{h}^j)[\hat{S}^{\bar{h}}_t-\hat{S}^h_t]^j\geq 0.$$
One can easily check that the same is true for $0<h^j\leq H^j_t$. 

\bigskip

Suppose now we are in case \ref{item2dimFTAP}. Recall that $\hat{S}_t^\xi\in L(\omega)$.  For any $h\in F$ there exists $\varepsilon>0$ such that for any  $s \in \mathbb{S}_{t+1}(\Sigma_{t}^{\omega})$,
$$h\cdot (s-\hat{S}_t^h)+h\cdot (\hat{S}_t^h-\hat{S}_t^\xi)\geq\varepsilon\Longrightarrow h\cdot (s-\hat{S}_t^h)\geq \varepsilon-h\cdot (\hat{S}_t^h-\hat{S}_t^\xi).$$
There exists $\alpha_2>0$ such that $\alpha_2(\varepsilon-h\cdot (\hat{S}_t^h-\hat{S}_t^\xi))\geq -\delta/2^{t}$. Denote
\begin{equation}\label{coeffMeas2}\alpha_2:=\min\left\{\dfrac{\delta}{2^{t}|\varepsilon-h\cdot (\hat{S}_t^h-\hat{S}_t^\xi)|},1\right\}\qquad \bar{h}:=\alpha_2 h\in F.
\end{equation}
Similarly as above if $h^j\leq 0$ then from \eqref{extension} and $\alpha_2\leq 1$ we get $[\hat{S}^h_t]^j\leq [\hat{S}^{\bar{h}}_t]^j$ and, analogously, $[\hat{S}^h_t]^j\geq [\hat{S}^{\bar{h}}_t]^j$ when $h^j\geq 0$. We thus get $\bar{h}\cdot (\hat{S}_t^h-\hat{S}_t^{\bar{h}})\geq 0$ and hence
$$\bar{h}\cdot (s-\hat{S}_t^{\bar{h}})=\bar{h}\cdot (s-\hat{S}_t^{h})+\bar{h}\cdot (\hat{S}_t^h-\hat{S}_t^{\bar{h}})\geq\bar{h}\cdot (s-\hat{S}_t^{h})\geq-\delta/2^{t}.$$ Observe now that in case \ref{item2dimFTAP}., $V_{t-1}(H)+H_t\cdot\hat{S}^{\bar{h}}_t\geq \delta/2^{t-1}$ and hence
$$V_{t-1}(H)+H_t\cdot\hat{S}^{\bar{h}}_t+\bar{h}\cdot (s-\hat{S}_t^{\bar{h}})\geq \delta/2^t$$
as desired.

\end{proof}
\begin{remark}\label{interpretationAB}
The sets $A_t$ and $B_t$ represents two different actions that must be taken in order to obtain a Model Independent Arbitrage. Note indeed that $A_t\cap B_t=\varnothing$. Fix $\omega\in\Omega$ and $t$. If $\omega\in A_t$, a new position is open. No strategy has been open before $t$ since we are restricting to the set $\bigcap_{u=0}^{t}\{H_u=0\}$ and, since $\tau(\omega)=t$, this is the first time that the market offers the possibility of a sure gain by trading in $S$ (see \eqref{ArbAt}). At this stage we are not concerned about liquidating the position. Suppose that at time $t$ we already have an open position (so $\omega\in A_u$ for some $u\leq t$). If $\omega\notin B_t$ then it can be liquidated at this time, since $H_{t+1}=0$ is admissible, and we obtain a strictly positive wealth with zero initial cost by \eqref{liquidation}. If $\omega\in B_t$ then it is not possible to liquidate the position at this time and we need to keep (or modify) the position and close it at subsequent times. By noting that $B_T$ is always the empty set, either because the position is closed before $T$ or because $\{\hat{S}^{sign H_T}_T\notin \mathbb{S}_T\}=\varnothing$ on $\{H_T\neq 0\}$ we see, by \eqref{liquidation}, that it is always possible to close the position opened on $A_u$ with a positive gain.
\end{remark}

\section{On Superhedging}\label{sec_super}
Recall the definition of the class $\mathcal{M}_{\overline{\Pi}}$ of price systems consistent with the bid-ask spread $\Pi$ (see Definition \ref{defCPS}) and the definition of $C_t$ in \eqref{convexBidAsk}. Consider the following
\begin{equation}\label{projCPS}
\mathcal{Q}:=\left\lbrace Q\in\mathfrak{P}\mid \exists S=(S_t)_{t\in I} \textrm{ with } S_t\in\mathcal{L}^0(\F_t; C_t) \textrm{ which is a $Q$-martingale} \right\rbrace,
\end{equation}
or, in other words, the projection of $\mathcal{M}_{\overline{\Pi}}$ on the set of probability measures and
\begin{equation}\label{processes}
\mathcal{S}:=\left\{S=(S_t)_{t\in I} \mid S_t\in\mathcal{L}^0(\mathcal{F}_t; C_t)\textrm{ and }\exists Q\in\mathfrak{P}\textrm{ s.t. }S \textrm{ is a }Q\textrm{-martingale}\right\},
\end{equation}
namely, the projection of $\mathcal{M}_{\overline{\Pi}}$ on the set of $\mathbb{F}$-adapted process. For any $S\in\mathcal{S}$ define also the section of $\mathcal{M}_{\overline{\Pi}}$ as 
\begin{equation}\label{sectM}
\mathcal{Q}_{S}:=\left\{Q\in \mathcal{Q}\mid S \textrm{ is a }Q\textrm{-martingale}\right\}.
\end{equation}
The maximal $\mathcal{Q}_{S}$-polar set has been characterized in \cite{BFM14} and denoted as $(\Omega_*(S))^c$. In particular $\Omega_*(S)=\{\omega\in\Omega\mid  \exists Q\in\mathcal{Q}_S \text{ such that }Q(\{\omega\})>0\}$. We here adapt the definition of $\Omega_*$ in this market with frictions. 
\begin{definition}\label{Omegastar} Let $\mathcal{Q}$ as in \eqref{projCPS}. We define the efficient support of the family of consistent price systems $\mathcal{M}_{\overline{\Pi}}$ as
$$\Omega_*:=\left\{\omega\in\Omega\mid \exists Q\in\mathcal{Q} \text{ such that }Q(\{\omega\})>0\right\}.$$
\end{definition}

For convenience of the reader we here recall the expression of the value process of a strategy $H$ from equation \eqref{valueFunction}, namely,
\begin{equation}\label{valueStrategy} V_T(H)=\sum_{t=0}^T\sum_{j=1}^d\left(H^j_t-H^j_{t+1}\right)\left(\oS^j_t\mathbf{1}_{\{H^j_t\leq H^j_{t+1}\}}+\uS^j_t\mathbf{1}_{\{H^j_{t+1}\leq H^j_{t}\}}\right).
\end{equation}

The aim of this section is to prove the following version of the superhedging Theorem:

\begin{theorem}\label{ThmSuper} Let $g:\Omega\mapsto\mathbb{R}$ be $\mathcal{F}_T$-measurable
\begin{equation}\label{superHmulti}
\sup_{Q\in\mathcal{Q}}\mathbb{E}_Q[g]=\inf\{x\in\mathbb{R}\mid \exists H\in\mathcal{H}\textrm{ s.t. }x+V_T(H)\geq g\quad \forall\omega\in\Omega_*\}=:\overline{p}(g)
\end{equation}where $\mathcal{Q}$ is defined in \eqref{projCPS} and $\Omega_*$ in Definition \ref{Omegastar}.
\end{theorem}
\begin{proof}[Proof of $(\leq)$] Assume $\mathcal{M}_{\overline{\Pi}}\neq\varnothing$ otherwise is trivial.
Let $S=(S_t)_{t\in I}$ be a process in $\mathcal{S}$. Take $x\in\R$, $H\in\mathcal{H}$ such that $x+V_T(H)\geq g$. For any strategy $H$, and for any $S\in\mathcal{S}$, inequality \eqref{frictionleassDominance} implies that $E_Q[V_T(H)]\leq 0$ respect to any martingale measure $Q$ for the process $S$. Since this is true for an arbitrary couple $(S,Q)$ and by recalling that $\Omega_*$ is the efficient support of the consistent price system (see Definition \ref{Omegastar}) we have $$g(\omega)\leq x+V_T(H)(\omega)\quad\forall \omega\in\Omega_*\quad\Longrightarrow\quad E_Q[g]\leq x\quad \forall Q\in \mathcal{Q}.$$ Take now the supremum over $Q\in\mathcal{Q}$ and then the infimum over $x\in\R$ in both sides to obtain
$$\sup_{Q\in\mathcal{Q}}\mathbb{E}_Q[g]\leq\overline{p}(g)$$ as desired.
\end{proof}

As usual one implication is easy. In order to prove the opposite we need some preliminary results.\\ 
We will construct now an auxiliary superhedging problem which involves a family of processes in $\mathcal{S}$, where $\mathcal{S}$ is defined in \eqref{processes}.\\

Introduce first,
\begin{equation}
F_T:\Omega\times \R^d\mapsto\R\quad \text{ defined as }\quad F_T(\omega,x)=g(\omega)\quad\forall \omega\in\Omega,\ x\in\R^d.
\end{equation}
Recall that, starting with $\mathbb{S}_{T+1}(\omega):=\R^d$, the random set \begin{equation}\label{IV}
\mathbb{S}_t(\omega):=\overline{conv}\{\mathbb{S}_{t+1}(\widetilde{\omega})\mid \widetilde{\omega}\in\Sigma_{t}^\omega\}\cap C_t
\end{equation}
is \Ftt-measurable for every $t=T,\ldots,0$, from Lemma \ref{LemmaTheta_meas}.
\begin{definition}\label{defF} We call the $\mathbb{S}$-superhedging problem the following backward procedure.
For any $t=T,\ldots,1$, for any $y\in\R$, define
$$\mathcal{H}_t^y(\omega,x)=\left\{H\in \R^d\mid  y+H\cdot(s-x)\geq F_t(\omega,s)\quad \forall s\in \mathbb{S}_t(\widetilde{\omega}),\ \forall \widetilde{\omega
}\in \Sigma _{t-1}^{\omega }\right\}$$
and set $$F_{t-1}(\omega,x):=\inf\left\{y\in\R\mid\mathcal{H}_t^y(\omega,x)\neq\varnothing \right\}.$$
We simply denote by $\mathcal{H}_t(\omega,x):=\mathcal{H}_t^{F_{t-1}(\omega,x)}(\omega,x)$ the set of optimal strategies at time $t\in I$ and by $\mathcal{A}_t(\omega,x):=\cup_{y\in\R}\{y\}\times\mathcal{H}_t^y(\omega,x)$ the set of acceptable couples. Both might be empty.\\
$F_0(x_0)$ will be called the $\mathbb{S}$-superhedging price for the initial value $x_0\in \R^d$.

\end{definition}

The next Proposition is crucial for the well-posedness of the prescribe procedure. It provides fundamental measurability properties for the whole scheme. Its proof is technical, as well as the proof of the subsequent results, and hence they are all postponed to Section \ref{Secproofs}.\\ Recall that a function $F:\Omega\times \R^d\mapsto \R\cup\{\pm\infty\}$ is called a Carath\'eodory map if: i) $F(\omega,x)$ is continuous in $x$, for every $\omega$ fixed, and ii) $F(\omega,x)$ is measurable in $\omega$, for every $x$ fixed.

\begin{proposition}\label{Fmeascont}
Let $F_t(\cdot,\cdot):\Omega\times \R^d\mapsto \R\cup\{\pm\infty\}$ for $t=0,\ldots T$ as in Definition \ref{defF}.
Denote by $D_{F_t}(\omega):=\{x\in\R^d\mid F_t(\omega,x)> -\infty\}$ the effective domain.\\
We have that
\begin{enumerate}
\item\label{item1cont} For every $x\in\R^d$ fixed, the map $F_t(\cdot,x)$ is \Ftt-measurable.\\ Moreover, when finite, $F_t(\cdot,x)$ is a minimum.
\item \label{item3cont}For every $\omega\in\Omega$ the map $F_t(\omega,\cdot)$ restricted to $\overline{D_{F_t}(\omega)}$ is continuous.
\item\label{item2cont} For every $\omega\in\Omega$, $\overline{D_{F_t}(\omega)}$ is convex.
\end{enumerate}
Items \ref{item1cont} and \ref{item3cont} imply that $F_t(\cdot,\cdot)$ is a Carath\'eodory map in its effective domain.
\end{proposition}
\begin{proof}
We postpone the proof to Section \ref{Secproofs}.
\end{proof}

\bigskip

For any initial value $x_0\in\R$ the $\mathbb{S}$-superhedging price $F_0(x_0)$, from Definition \ref{defF}, represents (when finite) the minimum amount of cash needed for superhedging $F_t(\omega,s)$, for any time $t\in I$, for any $\omega\in\Omega$ and for any $s\in \mathbb{S}_t(\widetilde{\omega})$. This value looks too conservative since it consider many possible values in the bid-ask spread for $S_{t}$. We nevertheless show the existence of $\bar{x}_0\in C_0$ such that: i) there exists a process $(S_t)_{t\in I}$ with $S_0=\bar{x}_0$ and with values in the bid-ask spread such that the superhedging price of $g$ with no frictions is $F_0(\bar{x}_0)$. ii) there exist a family of random vectors, provided by the solution of the $\mathbb{S}$-superhedging problem, which compose a self-financing trading strategy satisfying $$ F_0(\bar{x}_0)+V_T(H)\geq g\qquad\forall\omega\in\Omega_*.$$

We prove this in a constructing way. More precisely we need the following step-forward iteration: suppose that at time $t\geq 1$ the random variables $S_{t-1}\in\Lt$ and $H_t\in\mathcal{L}^0(\Ft;\R^d)$ with $H_t(\omega)\in\mathcal{H}_t(\omega,S_{t-1}(\omega))$ for every $\omega\in\Omega$, are given and define
\begin{equation}\label{X}
X_{t-1}(\omega):=F_{t-1}(\omega,S_{t-1}(\omega)).
\end{equation}

\begin{lemma}\label{Sprocess} Suppose $X_{t-1}(\omega)<\infty$ for any $\omega\in\Omega$. There exists a random vector $S_t\in\mathcal{L}^0(\Ftt;C_t)$ such that, for all $\omega\in\Omega$,
$$X_{t-1}(\omega)=\inf\{y\in\R\mid\exists H\in \R^d \text{ s.t. }y+H\cdot\Delta S_{t}(\widetilde{\omega})\geq F_t(\widetilde{\omega},S_{t}(\widetilde{\omega}))\quad\forall \widetilde{\omega}\in\Sigma_{t-1}^\omega\}$$
where $\Delta S:=S_{t}-S_{t-1}$. Moreover, if $X_{t-1}(\omega)>-\infty$, $H_t(\omega)$ is an optimal strategy.
\end{lemma}
\begin{proof}We postpone the proof to Section \ref{Secproofs}.
\end{proof}

We use Lemma \ref{Sprocess} as a building block for the desired process in i): the next Proposition shows that it is possible to construct a frictionless process whose superhedging price coincides with $F_0(x_0)$.

\begin{proposition}\label{propSprocess} For every $x_0\in C_0$ there exists a price process $S=(S_t)_{t\in I}$ such that:\begin{itemize}
\item $S_0=x_0$, $S_t\in\mathcal{L}^0(\F_t;C_t)$ for every $0\leq t\leq T$.
\item Let $\mathcal{H}^{pred}$ the class of $\mathbb{F}$-predictable process. Then, $$\inf\left\{x\in\R\mid\exists H\in\mathcal{H}^{pred}\text{ s.t. }x+(H\circ S)_T(\omega)\geq g(\omega) \quad\forall\omega\in\Omega_*(S)\right\}=F_0(x_0)$$
where $\Omega_*(S):=\{\omega\in\Omega\mid \exists Q\in\mathcal{Q}_S \text{ s.t. }Q(\{\omega\})>0\}$ and $\mathcal{Q}_S$ is defined in \eqref{sectM}.
\end{itemize}
\end{proposition}
\begin{proof}We postpone the proof to Section \ref{Secproofs}.
\end{proof}

We now construct, for a given initial value $x_0\in C_0$, a strategy $H:=(H_1,\ldots H_T)$ whose terminal payoff, considering transaction costs, dominates $g$. We again first show a one-step iteration. Recall from Definition \ref{defF} that $\mathcal{H}_{t+1}(\cdot,\cdot)$ is the set of optimal strategies for the (conditional) $\mathbb{S}$-superhedging problem.

\begin{proposition}\label{S2process}
 There exist a random vector $\widehat{S}_t\in\mathcal{L}^0_t(\mathcal{F}_t; C_t)$ and a trading strategy $H_{t+1}\in\Ltt$ such that, for every $\omega\in\{X_{t-1}<\infty\}$,
\begin{equation}
 X_{t-1}(\omega)+H_{t}\cdot(\widehat{S}_{t}(\widetilde{\omega})-S_{t-1}(\widetilde{\omega}))\geq F_t(\widetilde{\omega},\widehat{S}_{t}(\widetilde{\omega}))\quad\forall \widetilde{\omega}\in\Sigma_{t-1}^\omega.
\end{equation}
 Moreover $H_{t+1}(\omega)\in\mathcal{H}_{t+1}(\omega,\widehat{S}_t(\omega))$ and the following properties are satisfied:
\begin{itemize}
\item if $H^i_t(\omega)< H^i_{t+1}(\omega)$ then $\widehat{S}^i_t(\omega)=\oS^i(\omega)$
\item if $H^i_t(\omega)> H^i_{t+1}(\omega)$ then $\widehat{S}^i_t(\omega)=\uS^i(\omega)$
\end{itemize}
In particular if $\hat{S}^i_t\in \left(\uS^i(\omega),\oS^i(\omega)\right)$ we necessarily have $H^i_t(\omega)=H^i_{t+1}(\omega)$.
\end{proposition}
\begin{proof}
We postpone the proof to Section \ref{Secproofs}.
\end{proof}
\begin{remark} With a slight abuse of notation, when $X_{t-1}(\omega)=-\infty$ we intend that there exists a sequence $\{(y_n,H_n)\}\subseteq\R\times\Ltt$ with $y_n\rightarrow-\infty$, such that for every $n\in\mathbb{N}$ the conditions of Proposition \ref{S2process} are satisfied. The same apply to Corollary \ref{corS2process} when $F_0(x_0)=-\infty$.
\end{remark}

\begin{corollary}\label{corS2process} For every $x_0\in C_0$ with $F_0(x_0)<\infty$ there exists a predictable process $H:=(H_1,\ldots H_T)$ such that
\begin{equation*} F_0(x_0)+(0-H_1\cdot x_0)+\sum_{t=1}^T\sum_{j=1}^d\left(H^j_t-H^j_{t+1}\right)\left(\oS^j_t\mathbf{1}_{\{H^j_t\leq H^j_{t+1}\}}+\uS^j_t\mathbf{1}_{\{H^j_{t+1}< H^j_{t}\}}\right)\geq g \text{ on }\Omega_*.
\end{equation*}
\end{corollary}
\begin{proof}
We postpone the proof to Section \ref{Secproofs}.
\end{proof}

\begin{remark}\label{remarkInf}
Observe that if $F_0(x_0)=-\infty$ then, from \eqref{frictionleassDominance}, the superhedging problem for any frictionless process $S=(S_t)_{t\in I}$ with $S_0=x_0$ has solution $-\infty$, from which $\mathcal{Q}_S=\varnothing$. 
\end{remark}
We can now conclude the proof of Theorem \ref{ThmSuper} as follows:

\begin{proof}[Proof of $(\geq)$ in \eqref{superHmulti} of Theorem \ref{ThmSuper}]
Let $F_0(x)$ be the solution of the superhedging problem in Definition \ref{defF}. Take $$m:=\sup_{x\in C_0} F_0(x).$$ Suppose first that $m=\infty$. There exists a sequence $x_n\in C_0$ such that $F_0(x_n)\rightarrow \infty$. From Proposition \ref{propSprocess} there exists a sequence of processes $S^n:=(S^n_t)_{t\in I}\subseteq \mathcal{S}$ whose (frictionless) superhedging price explode to $\infty$ and hence the inequality is trivial. If $m=-\infty$ then by Corollary \ref{corS2process} and \eqref{frictionleassDominance} the equality follows again trivially as a degenerate case: $\Omega_*=\varnothing$ (see Remark \ref{remarkInf}). If $m$ is finite then $m=\sup_{x\in \overline{D_{F_0}}} F_0(x)$. By Proposition \ref{Fmeascont} $F_0$ is non-random, continuous and $\overline{D_{F_0}}$ is a closed subset of a compact set $C_0$. Thus $m$ is a maximum and we denote by $\bar{x}_0$ a maximizer. By Proposition \ref{propSprocess} there exists a process $S:=(S_t)_{t\in I}$ with $S_0=\bar{x}_0$ whose superhedging price is $m$, namely,
\begin{equation}\label{finalStep1}
m=\inf\left\{x\in\R\mid\exists H\in\mathcal{H}\text{ s.t. }x+(H\circ S)_T(\omega)\geq g(\omega) \quad\forall\omega\in\Omega_*(S)\right\}=\sup_{Q\in\mathcal{Q}_{S}}E_Q[g]
\end{equation}
where the last equality derives from Theorem 1.1 in \cite{BFM15}.\\
On the other hand by adding a fictitious node $t=-1$ to the $\mathbb{S}$-superhedging problem in Definition \ref{defF}, with $\mathbb{S}_{-1}=\bar{x}_0$, we have that the minimization
$$\inf\left\{y\in\R\mid H\in \R^d\text{ s.t. }y+H\cdot(s-\bar{x}_0)\geq F_0(s)\quad \forall s\in \mathbb{S}_0,\right\}$$

has the obvious solution $X_{-1}=m$, with corresponding optimal strategy $H_{0}=0$. By applying Proposition \ref{S2process} we obtain $\hat{S}_0=\bar{x}_0$ (see also \eqref{argmin}) and $H_1$ such that $$H_1\cdot \bar{x}_0=\sum_{j=1}^d H^j_{1}\left(\oS^j_0\mathbf{1}_{\{0\leq H^j_{1}\}}+\uS^j_0\mathbf{1}_{\{H^j_{1}< 0\}}\right).$$ Apply now Corollary \ref{corS2process}, with $x_0=\bar{x}_0$, to get the existence of a trading strategy $(H_t)_{t\in I}$ such that (cfr equation \eqref{valueStrategy})
\begin{equation}\label{finalStep2}
m+V_T(H)(\omega)\geq g(\omega)\quad\forall \omega\in\Omega_*.
\end{equation}

The desired inequality follows from \eqref{finalStep1} and \eqref{finalStep2}:
\begin{equation*}
\sup_{Q\in\mathcal{Q}}\mathbb{E}_Q[g]= \sup_{\widetilde{S}\in\mathcal{S}}\sup_{Q\in\mathcal{Q}_{\widetilde{S}}}\mathbb{E}_Q[g] \geq  \sup_{Q\in\mathcal{Q}_{S}}\mathbb{E}_Q[g]= m\geq \overline{p}(g).
\end{equation*}
\end{proof}

\subsection{Proofs}\label{sec_proofs}
\label{Secproofs}

\begin{remark}\label{comments}Let us point out two simple facts that we will often use in the following proofs.\\ 
First note that if, for some $\omega\in\Omega$, there exists $v\in\R^d$ and $\varepsilon>0$ such that $v\cdot (s-x)\geq\varepsilon$ for every $s\in\mathbb{S}_{t+1}(\widetilde{\omega})$ and for every $\widetilde{\omega}\in \Sigma_t^\omega$ then $F_t(\omega,x)=-\infty$ since for every acceptable couple $(y,H)\in\mathcal{A}_t(\omega,x)$ we have $(y-\alpha\varepsilon,H+\alpha v)\in\mathcal{A}_t(\omega,x)$, $\forall \alpha>0$.\\
Second note that if $F_t(\omega,x)=-\infty$ then there exists a sequence $\{(y_n,H_n)\}\subseteq\mathcal{A}_t(\omega,x)$ with $y_n\rightarrow-\infty$. From \eqref{frictionleassDominance}, for \emph{any} $S_t$ with values in the bid-ask spread, the same sequence satisfies $y_n+H_n\cdot(S_t(\widetilde{\omega})-x)\geq F_{t+1}(\widetilde{\omega},S_t(\widetilde{\omega}))$ for all $\widetilde{\omega}\in\Sigma_{t-1}^{\omega}$. The infimum over $y$ is again $-\infty$.\\
\end{remark}

\begin{proof}[Proof of Proposition \ref{Fmeascont}]
For $t=T$ the claim is trivial. Suppose it is true for all $t+1\leq u \leq T-1$.

\bigskip

\textbf{\ref{item1cont}.} 
We first show that $\mathbb{S}_{t+1}$ takes values in the closure of the effective domain of $F_{t+1}(\omega,s)$. For $t=T-1$ there is nothing to show. From \eqref{IV}, any $s\in \mathbb{S}_{t+1}(\omega)$ is limit of convex combinations of elements in $ \mathbb{S}_{t+2}(\Sigma^\omega_{t+1})$. Let $s_n\rightarrow s$. For any $n\in\mathbb{N}$, there exist, without loss of generality: 
\begin{itemize}
\item $\omega_1,\ldots,\omega_{k(n)}$ with $\omega_i\in\Sigma^\omega_{t+1}$ for every $i$;
\item $z_1,\ldots,z_{k(n)}$ with $z_i\in \mathbb{S}_{t+2}(\omega_i)$ for every $i$;
\item $\lambda_1,\ldots\lambda_{k(n)}$, with $0< \lambda_i< 1$ for every $i$;
\end{itemize}
such that $s_n:=\sum_{i=1}^{k(n)}\lambda_i z_i$. Consider a frictionless, one-period model, on $\{z_1,\ldots,z_n\}$ with $S_0=s_n$, $S_1(z_i)=z_i$ for every $i$. $Q(\{z_i\}):=\lambda_i$ define a martingale measure for the process $S$. 

Denote by $\mathcal{M}(S)$ the set of martingale measures for $S$ and $\overline{p}_{S}(g)$ the (frictionless) superhedging price for $g(z_i):=F_{t+2}(\omega_i,z_i)$ in the one-period model.
From the classical theory $$-\infty<\sum_{i=1}^{k(n)}\lambda_i g(z_i)\leq\sup_{Q\in\mathcal{M}(S)}E_Q[g]=\overline{p}_{S}(g)\leq F_{t+1}(\omega, s_n) $$ where the last inequality follows from $F_{t+1}$ being the solution of the (conditional) $\mathbb{S}$-superhedging problem.
We thus have that $s_n\in D_{F_{t+1}}(\omega)$ for every $n$ and hence $s\in \overline{D_{F_{t+1}}(\omega)}$.\\

Observe now that, from the inductive hypothesis, $F_{t+1}$ is a Carath\'eodory map in its domain and since $\mathbb{S}_{t+1}$ takes value in $\overline{D_{F_{t+1}}}$ we can apply Corollary \ref{corSuperHmultif} in the Appendix with $u=t+1$, $X_{u-1}=x$, $X_u=\mathbb{S}_{u}$, $C=\R^d$, to get the measurability of
\begin{equation*}
A_C(\omega )=\left\{  (H,y)\in \R^{d+1}\mid y+H\cdot(s-x)\geq F_{t+1}(\omega,s)\quad \forall s\in \mathbb{S}_{t+1}(\widetilde{\omega}),\ \forall \widetilde{\omega
}\in \Sigma _{t}^{\omega }\right\}.
\end{equation*}
The measurable map $M_{t}$ from Corollary \ref{corSuperHmultif} represents, for any $\omega\in\Omega$ the minimum amount of cash needed for superhedging $F_{t+1}(\omega,s)$ for any $s\in \mathbb{S}_{t+1}(\widetilde{\omega})$, and hence it correspond to $F_t(\cdot,x)$.

\bigskip

\textbf{\ref{item2cont}.} We first show item \ref{item2cont}.\\
Fix $\omega\in\Omega$. If $D_{F_t}=\varnothing$ there is nothing to show. Denote by $$A(x):=\{H\in\R^d\mid H\cdot (s-x)\geq 0 \quad \forall s\in\mathbb{S}_{t+1}(\Sigma_{t}^\omega)\ \text{ with $> 0$ for some } \bar{s}\}.$$
We show that
the set $C:=\{x\in D_{F_t}(\omega)\mid A(x)=\varnothing\}$ is convex and $\overline{D_{F_t}(\omega)}=\overline{C}$ from which the thesis follows. Denote by $$\Gamma:=conv\{\mathbb{S}_{t+1}(\Sigma_t^\omega)\}.$$ Take now $x_1,x_2\in C$ and recall that, from Hyperplane separation Theorem, $A(x_i)=\varnothing$ if and only if $x_i\in ri(\Gamma)$. As $\Gamma$ is a convex set for any $0\leq \lambda\leq 1$, $\lambda x_1+(1-\lambda)x_2\in ri(\Gamma)$ and hence $A(\lambda x_1+(1-\lambda)x_2))=\varnothing$ from which $C$ is convex.

We now show that if $x\in D_{F_t}(\omega)$ then there exists a sequence $x_k\in C$ such that $x_k\rightarrow x$. Take $x\notin C$ otherwise is trivial. Note first that $x\in \overline{\Gamma}$ otherwise by Hyperplane separation Theorem there would exists $v\in\R^d$ and $\varepsilon>0$ such that $v\cdot (\mathbb{S}_{t+1}(\Sigma_t^\omega)-x)\geq\varepsilon$ which would give $x\notin D_{F_t}(\omega)$ (see also Remark \ref{comments}).\\
Take now $\widetilde{x}\in ri(\Gamma)$, for every $k\in\mathbb{N}$ set
$$x_k:=\left(1-\frac{1}{k}\right) x+\frac{\widetilde{x}}{k}\in ri(\Gamma)$$
clearly $x_k\rightarrow x$ as $k\rightarrow\infty$ and again from Hyperplane separation Theorem $x_k\in C$.\\

\bigskip

\textbf{\ref{item3cont}.} 
First observe that if there exists $\widetilde{x}$ such that $F_t(\omega,\widetilde{x})=+\infty$ then $F_t(\omega,\cdot)\equiv+\infty$ and hence: $D_{F_t}(\omega)=\R^d$ and $F_t(\omega,\cdot)$ is trivially continuous. Indeed, since $F_t(\omega,\widetilde{x})=+\infty$, for any $H\in\R^d$ there exists a sequence $\{(\omega_n,s_n)\}_{n\in\mathbb{N}}$ such that $H\cdot(s_n-\widetilde{x})-F_{t+1}(\omega_n,s_n)\rightarrow-\infty$. Therefore the same holds for the sequence $H\cdot(s_n-x)-F_{t+1}(\omega_n,s_n)$ with $x$ arbitrary. Thus, $F_t(\omega,x)=+\infty$. \\

We may now suppose that $F_t(\omega,\cdot)<+\infty$. We first show that $F(\omega,\cdot)$ is upper semi-continuous at $x\in\overline{D_{F_t}(\omega)}$.\\

For $x\in D_{F_t}(\omega)$, Corollary \ref{corSuperHmultif} in the Appendix implies that there exists an optimal strategy $H$ such that 
\begin{equation}\label{bestIneq}
F_t(\omega,x)+H\cdot(s-x)\geq F_{t+1}(\widetilde{\omega},s)\quad \forall s\in \mathbb{S}_{t+1}(\widetilde{\omega}),\ \forall \widetilde{\omega
}\in \Sigma _{t}^{\omega }.
\end{equation}
Let now $\{x_k\}_{k=1}^{\infty}$ such that $x_k\rightarrow x$ for $k\rightarrow\infty$. Observing that
$H\cdot(s-x)=H\cdot(s-x_k)+H\cdot(x_k-x)$
we get, from \eqref{bestIneq}, $F_t(\omega,x_k)\leq F_t(\omega,x)+H\cdot(x_k-x)$. By taking limits in both sides we can conclude that $F_t(\omega,\cdot)$ is upper semi-continuous:\\
\begin{equation}\label{uscontinuityF}
\limsup_{k\rightarrow\infty}F_t(\omega,x_k)\leq F_t(\omega,x).
\end{equation}
The case of $x\notin D_{F_t}(\omega)$ is similar. Since $F_t(\omega,x)=-\infty$ there exists a sequence $\{H_n\}$ such that \eqref{bestIneq} is satisfied with $(-n,H_n)$ replacing $(F_t(\omega,x),H)$. We analogously obtain $F_t(\omega,x_k)\leq -n+H_n\cdot(x_k-x)$. By taking the limit in $k$ in both sides we get $\limsup_{k\rightarrow\infty}F_t(\omega,x_k)\leq -n$ for any $n\in\mathbb{N}$, from which the upper semi-continuity follows.

\bigskip

We now turn to the lower semi-continuity. Let $x\in\overline{D_{F_t}(\omega)}$. If $x\notin D_{F_t}(\omega)$, that is, $F_t(\omega,x)=-\infty$, from the previous step we already have continuity. Suppose therefore $x\in D_{F_t}(\omega)$ and let $H$ an optimal strategy such that \eqref{bestIneq} is satisfied.

\bigskip

\textbf{case a)} If the inequality in \eqref{bestIneq} is actually an equality we have perfect replication and we can infer that for any $\widetilde{x}\in D_{F_t}(\omega)$ we have $F_t(\omega,\widetilde{x})= F_t(\omega,x)+H\cdot(\widetilde{x}-x)$. Indeed, observe first that by adding and subtracting $H\cdot(\widetilde{x}-x)$ in \eqref{bestIneq}, which holds with equality by assumption, we obtain
$$F_t(\omega,x)+H\cdot(\widetilde{x}-x)+H\cdot(s-\widetilde{x})= F_{t+1}(\widetilde{\omega},s)\quad \forall s\in \mathbb{S}_{t+1}(\widetilde{\omega}),\ \forall \widetilde{\omega
}\in \Sigma _{t}^{\omega },$$
from which, $F_t(\omega,\widetilde{x})\leq F_t(\omega,x)+H\cdot(\widetilde{x}-x)$. Suppose now that there exists a cheaper superhedging strategy $H_z$ with cost $z\in\R$. Namely, $(z,H_z)$ satisfies $l:=z-F_t(\omega,x)+H\cdot(\widetilde{x}-x)<0$ and
$$z+H_z\cdot(s-\widetilde{x})\geq F_{t+1}(\widetilde{\omega},s)\quad \forall s\in \mathbb{S}_{t+1}(\widetilde{\omega}),\ \forall \widetilde{\omega
}\in \Sigma _{t}^{\omega }.$$ By subtracting the previous equality we obtain $$(H_z-H)\cdot (\mathbb{S}_{t+1}(\Sigma_t^\omega)-\widetilde{x})\geq-l>0$$ from which $\widetilde{x}\notin D_{F_t}(\omega)$ (see also Remark \ref{comments}) and thus a contradiction.\\ By considering $\{x_k\}_{k=1}^{\infty}$ such that $x_k\rightarrow x$ we obtain $$\lim_{k\rightarrow\infty}F_t(\omega,x_k)=\lim_{k\rightarrow\infty}(F_t(\omega,x)+H\cdot(x_k-x))=F_t(\omega,x)$$
as desired.

\bigskip

\textbf{case b)}
Define \begin{equation}\label{eqG}
G_t(\omega,x):=\sup \left\{ y\in \mathbb{R}\mid \exists \,H\in \mathbb{R}%
^{d}:\;y+H\cdot (s-x)\leq F_{t+1}(\widetilde{\omega },s),\ \forall s\in \mathbb{S}_{t+1}(\widetilde{\omega}),\ \forall \widetilde{\omega
}\in \Sigma _{t}^{\omega }\right\}
\end{equation}%
and, for all $y\in\R$, the set
\begin{equation}\label{eqGamma}
\Gamma _{y}(x):= co\left(conv\left\{ \left[ s-x\ ;\ y-F_{t+1}(\widetilde{\omega },s)\right]\mid s\in \mathbb{S}_{t+1}(\widetilde{\omega}),\ \widetilde{\omega
}\in \Sigma _{t}^{\omega }\right\}\right)\subseteq \R^{d+1}.
\end{equation}
Note that $F_t(\omega,x)>G_t(\omega,x)$ otherwise there is perfect replication and we are back to case a).  Take $y\in (G_t(\omega,x),F_t(\omega,x))$ and note that necessarily $int(\Gamma _{y}(x))\neq \varnothing$ .
\\

If $0\in int(\Gamma _{y}(x))$ there exists $\bar{\varepsilon}>0$ such that for every $\varepsilon\leq\bar{\varepsilon}$, $B_{2\varepsilon}(0)\subseteq int(\Gamma_y(x))$. For any $z\in B_{\varepsilon}(0)$ of the form $z=(\widetilde{x},0)$ with $\widetilde{x}\in\R^d$, we have $0\in int(\Gamma _{y}(\widetilde{x}))$, hence, there is no non-zero $(H,h)\in \mathbb{R}%
^{d}\times \mathbb{R}$ , such that either
\begin{equation}\label{separationF}
h(y-F_{t+1}(\widetilde{\omega },s))+H\cdot (s-\widetilde{x})\geq 0\ \text{ or }\ h(y-F_{t+1}(\widetilde{\omega },s))+H\cdot (s-\widetilde{x})\leq 0
\end{equation}
is possible for every $s\in \mathbb{S}_{t+1}(\widetilde{\omega})$ and $\widetilde{\omega
}\in \Sigma _{t}^{\omega }$. In particular there is no $H\in
\mathbb{R}^{d}$ such that $y+H\cdot (s-\widetilde{x})\geq F_{t+1}(\widetilde{\omega },s)$ for every $s\in \mathbb{S}_{t+1}(\widetilde{\omega})$ and $\widetilde{\omega
}\in \Sigma _{t}^{\omega }$. Thus, $F_t(\omega,\widetilde{x})>y$. Since the same holds for every $\widetilde{x}$ such that $\|\widetilde{x}-x\|<\varepsilon$ with $\varepsilon$ arbitrary small, by considering a sequence $\{x_k\}_{k=1}^{\infty}$ such that $x_k\rightarrow x$ we have obtained $\liminf_{k\rightarrow\infty}F_t(\omega,x_k)>y$ for every $y\in (G_t(\omega,x),F_t(\omega,x))$. By taking the supremum over $y$ we have
\begin{equation}\label{lscF}
\liminf F_t(\omega,x_k)\geq F_t(\omega,x)
\end{equation} 
as desired.

\bigskip
If $0\notin int(\Gamma _{y}(x))$ there exists a separator $(H,h)\in \mathbb{R}%
^{d}\times \mathbb{R}$ such that \eqref{separationF} holds but since $y\in (G_t(\omega,x),F_t(\omega,x))$ we necessarily have $h=0$. Consider now a separator $\hat{H}:=(H,0)$ with $H\in\R^d$ and denote by $\hat{H}^{++}$, $\hat{H}^{+}$ the positive and non-negative half-spaces associated to $\hat{H}$. Analogously $\hat{H}^{--}$, $\hat{H}^{-}$. Define 
\begin{equation}\label{eqA}
A:=\{z\in\R^{d+1}\mid \hat{H}\cdot z=0\}\cap \overline{\Gamma_y(x)}.
\end{equation}

Observe that since $\Gamma _{y}(x)\subseteq \hat{H}^+$ and $0\in ri(A)$ from Lemma \ref{riA},  there exists $\bar{\varepsilon}>0$ such that for every $\varepsilon\leq\bar{\varepsilon}$, we have $B_{2\varepsilon}(0)\cap \hat{H}^{++}\subseteq int(\Gamma _{y}(x))$. As in case a) for every $z\in B_{\varepsilon}(0)\cap\hat{H}^{++}$ of the form $z=(\widetilde{x},0)$ we have $0\in int (\Gamma _{y}(\widetilde{x}))$. This implies $F_t(\omega,\widetilde{x})>y$. In order to conclude observe that if $(\widetilde{x},0)\in B_{\varepsilon}(0)\cap \hat{H}^-$ then $\widetilde{x}\notin ri(D_{F_t}(\omega))$. If indeed $\widetilde{x}$ is such that $H\cdot (\widetilde{x}-x)\leq 0$ then \begin{equation}\label{sep}
H\cdot (s-\widetilde{x})\geq 0\qquad \forall s\in \mathbb{S}_{t+1}(\widetilde{\omega}),\ \widetilde{\omega
}\in \Sigma _{t}^{\omega }.
\end{equation}
It is easy to see that in every neighbourhood of $\widetilde{x}$ there exists an element $\bar{x}$ for which, replacing $\widetilde{x}$ with $\bar{x}$ in \eqref{sep} the inequality is satisfied with a lower bound. Thus $\widetilde{x}$ is not in $D_{F_t}(\omega)$ (see also Remark \ref{comments}). 
\bigskip

We have therefore obtained that if a sequence $\{x_k\}_{k=1}^{\infty}\subseteq ri(D_{F_t}(\omega))$ satisfies $x_k\rightarrow x$ then \eqref{lscF} holds and hence, also in case b), the thesis.
\end{proof}

\begin{lemma}\label{riA}
Let $\hat{H},x\in\R^d$, $\omega\in\Omega$ be given. Let $G_t(\omega,x)$ and $\Gamma_y(x)$ from \eqref{eqG} and \eqref{eqGamma} respectively with $y\in (G_t(\omega,x),F_t(\omega,x))$. Let $A$ from \eqref{eqA}, then $0\in ri(A)$.
\end{lemma}
\begin{proof}
Suppose by contradiction that there exists $r\in\R^{d+1}$ such that $\hat{H}\cdot r=0$ and $\alpha r \notin A$ for every $\alpha>0$. Note that from $r\notin A$ we have $dist(r,\overline{\Gamma_y(x)})>0$ so that there exists $\delta>0$ such that $B_\delta(r) \cap \overline{\Gamma_y(x)}=\varnothing$. Since $\overline{\Gamma_y(x)}$ is a cone we can conclude that the segment $[0,\widetilde{r}]$ with $\widetilde{r}\in B_\delta(r)$ has empty intersection with $\overline{\Gamma_y(x)}$. Since obviously $0\in \cup_{0\leq\alpha\leq 1}\alpha B_{\delta}(r)$ we can infer that there exists $(\widetilde{H},\widetilde{h})$ with $\widetilde{h}\neq 0$ such that $$\widetilde{h}(y-F_{t+1}(\widetilde{\omega},s))+\widetilde{H}\cdot (s-x)\geq 0\qquad \forall s\in \mathbb{S}_{t+1}(\widetilde{\omega}),\ \widetilde{\omega
}\in \Sigma _{t}^{\omega }\ ,$$ 
which is a contradiction since $y\in (G_t(\omega,x),F_t(\omega,x))$.

\end{proof}

\begin{remark}
Observe that from the proof of Proposition \ref{Fmeascont} we actually obtained that $F_t(\omega,\cdot)$ is upper semi-continuous in the whole space $\mathbb{R}^d$ and note only on $\overline{D_{F_t}(\omega)}$. Note, moreover, that for showing the lower semi-continuity one could argue that $F_t(\omega,x)\leq F_t(\omega,x_k)+H_k\cdot(x-x_k)$, where $H_k$ is an optimal strategy associated to $F_t(\omega,x_k)$, and then take the limit. Nevertheless in order to conclude that $F_t(\omega,\cdot)$ is lower semi-continuous we would need, for instance, that the sequence $\{H_k\}$ is bounded, which in general cannot be guaranteed.

\end{remark}

\begin{proof}[Proof of Lemma \ref{Sprocess}]
Since $S_{t-1}$ is given, simply denote by $\mathcal{H}_t$ the random set $\mathcal{H}_t(\cdot,S_{t-1}(\cdot))$ which is \Ft-measurable as it coincides with $\mathcal{H}^M_u$ from Corollary \ref{corSuperHmultif} in the Appendix with $u=t$, $X_{u-1}=S_{u-1}$, $X_u=\mathbb{S}_u$, $C=\R^d$.\\ Note that on $\{X_{t-1}=-\infty\}$ the claim is trivial by \eqref{frictionleassDominance} (see also Remark \ref{comments}). Suppose therefore $X_{t-1}>-\infty$ which implies $\mathcal{H}_t\neq\varnothing$.
Define $$A_1(\omega):=\left\{(y,x)\in\R\times\R^{d}\mid (1,H)\cdot(y,x)=0\quad \forall H\in\mathcal{H}_t \right\},$$which is \Ft-measurable as it can be obtained as $(\{1\}\times\mathcal{H}_t)^*\cap -(\{1\}\times\mathcal{H}_t)^*$ (recall Notation \ref{eps-dual}). Define also
$$A_2(\omega):=\left\{\left(X_{t-1}(\omega)-F_{t}(\omega,s),\ s-S_{t-1}(\omega)\right)\mid s\in\overline{\mathbb{S}_t(\Sigma_{t-1}^{\omega})}\right\},$$
which is \Ftt-measurable being composition of the Carath\'eodory map $(\omega,x)\mapsto (X_{t-1}(\omega)-F_{t}(\omega,x),x-S_{t-1}(\omega))$ and the measurable random set $\overline{\mathbb{S}_t(\Sigma_{t-1}^{\omega})}$ (see also Corollary \ref{corTheta_Meas}). Define finally
$$A(\omega):=A_1(\omega)\cap A_2(\omega).$$
Every $a\in A$ is of the form $a=(X_{t-1}(\omega)-F_{t}(\omega,s),s-S_{t-1}(\omega))$ for some $s\in\overline{\mathbb{S}_t(\Sigma_{t-1}^{\omega})}$, and satisfies $X_{t-1}(\omega)+ H\cdot(s-S_{t-1}(\omega))=F_{t}(\omega,s)$ for every $H\in\mathcal{H}_t$. Note now that $A$ is closed-valued and $0\in ri(conv(A))$: if this is not the case then $A$ can be strictly separated from $\{0\}$ and $X_{t-1}(\omega)-F_{t}(\omega,s)+ \widetilde{H}\cdot(s-S_{t-1}(\omega))\geq \varepsilon>0$ for some $\widetilde{H}\in\R^d$ and hence $X_{t-1}$ is not optimal.\\
We now show that we can construct an \Ftt-measurable random vector $S_t$ such that the analogous set
$$A_2^S(\omega):=\left\{\left(X_{t-1}(\widetilde{\omega})-F_{t}(\widetilde{\omega},S_t(\widetilde{\omega})),\ \Delta S_t(\widetilde{\omega})\right)\mid \widetilde{\omega}\in \Sigma_{t-1}^{\omega}\right\},$$
where $\Delta S_t:= S_t-S_{t-1}$, satisfies $0\in ri(conv(A_1\cap A_2^S))$. For the same reason $X_{t-1}$ is the (conditional) frictionless superhedging price.

\bigskip

Take $\bar{d}:=d+1$ for simplicity of notation. We first extract an \Ftt-measurable collection $\{a_j\}_{j=1}^{2\bar{d}^2}\subseteq A$ and $\lambda:\Omega\mapsto\R^{2\bar{d}^2}$ \Ftt-measurable, such that,
\begin{equation}\label{conv0}
0=\sum_{j=1}^{n} \lambda_j(\omega)a_j(\omega)
\end{equation}
and $co(conv(\{a_j(\omega)\}_{j=1}^{2\bar{d}^2}))=lin A(\omega)$, which implies $0\in ri(\{a_j(\omega)\}_{j=1}^{2\bar{d}^2})$ (recall Notation \ref{notation}).\\ 
By denoting $\Delta^{2\bar{d}}$ the simplex in $\R^{2\bar{d}}$, define the function $L:\Omega\times\R^{\bar{d}\times 2\bar{d}}\times\Delta^{2\bar{d}}\mapsto\R^{\bar{d}}$ as 
$$L(\omega,x_1,\ldots,x_{2\bar{d}},\lambda):=\sum_{i=1}^{2\bar{d}}\lambda_i x_i\quad \text{for }\lambda\in\Delta^{2\bar{d}},\ x_i\in\R^{\bar{d}},\ \forall i=1,\ldots,2\bar{d}\ .$$ $L$ is a Carath\'eodory map since it does not depend on $\omega$ and is continuous in $(x_1,\ldots, x_{2\bar{d}},\lambda)$. Denote $A^ {2\bar{d}}$ the Cartesian product of $2\bar{d}$ copies of $A$ and $Y^1:=A^ {2\bar{d}}\times \Delta^{2\bar{d}}$. From Proposition \ref{preservation_measurability} in the Appendix, $Y^1$ is \Ftt-measurable and closed-valued. From the implicit map Theorem (Theorem \ref{implicit} with  $D(\omega)=\{0\}\subset\R^{\bar{d}}$) and from $0\in ri(conv (A))$ 
there exists $B^1:=\{a^1_1,\ldots,a^1_{2\bar{d}}\}$ and $\lambda^1:\Omega\mapsto\R^{2\bar{d}}$ \Ftt-measurable such that \eqref{conv0} is satisfied.\\

Note however that we might have $dim(B^1)<dim(A)$.  We iterate the process as follows. Suppose we are given $B^1,\ldots,B^{k-1}$ for $k\geq 2$. Consider the following closed-valued random set $$D^k(\omega):=lin\left\{a^i_j\mathbf{1}_{\lambda^i_j> 0}\mid j=1,\ldots,2\bar{d},\ i=1,\ldots, k-1\right\},$$ which is $\Ftt$-measurable by Proposition \ref{preservation_measurability}. Our aim is to find a set of vectors $B^k$ in $A\setminus D^{k}$ whose convex combination is in $D^{k}$. This implies that, together with the vectors in $B^1,\ldots,B^{k-1}$, they satisfy \eqref{conv0}. Let $B_{1/n}(0)$ be the open ball of radius $1/n$ with center in $0$. Since $A(\omega)\setminus D^{k}$ is not closed-valued, for any $n\in\mathbb{N}$, we define $A_n(\omega):=A(\omega)\setminus (D^{k}+B_{1/n}(0))$ which is closed-valued and measurable from Proposition \ref{preservation_measurability} and Lemma \ref{compl}. We define, moreover, $Y^k_n:=A_n^ {2\bar{d}}\times \Delta^{2\bar{d}}$ which is also $\Ftt$-measurable and closed-valued. Applying Theorem \ref{implicit} we obtain,
\begin{equation*}
E^k_n:=\left\{\omega\in\Omega\mid \exists y\in Y^k(\omega)\text{ with } L(\omega,y)\in D^k(\omega) \right\}
\end{equation*}
is $\Ftt$-measurable and there exists a measurable function $y_n:E^k_n\mapsto\R^{\bar{d}\times 2\bar{d}}$ such that
$$y_n(\omega)\in Y^k(\omega)\quad\text{and}\quad L(\omega,y_n(\omega))\in D^k(\omega)\quad \forall\omega\in E^k_n\ .$$
Note that for every $\omega\in\Omega$ there exist a finite number of elements whose convex combination belongs to $D^{k}$ or equivalently, there exists $n\in\mathbb{N}$ such that $\omega\in E^k_n$. We therefore have $\cup_{n\in\mathbb{N}}E^k_n=\Omega$ with $\hat{E}_n:=\cup_{i=1}^{n}E^k_i$ increasing in $n$. Thus $y:=\sum_{n\in\mathbb{N}}y_n\mathbf{1}_{\hat{E}_n\setminus \hat{E}_{n-1}}$ is well defined on $\Omega$. By taking $B^{k}:=\{a^k_1,\ldots,a^k_{2\bar{d}}\}$ the first $2\bar{d}$ components of $y$ in $\R^{\bar{d}}$ and $\lambda^k$ the last $\R^{2\bar{d}}$ component, we have
\begin{equation*}
\sum_{j=1}^{2\bar{d}} \lambda^k_j(\omega)a^k_j(\omega)\in D^k(\omega)
\end{equation*}
and hence \eqref{conv0} is satisfied for $B^1,\ldots,B^{k}$. Since, for every $k$, $B^k\cap lin(\cup_{i=1}^{k-1} B^j)=\varnothing$ we have that $dim (\cup_{i=1}^{k} B^j)$ is increasing in $k$ and therefore the procedure ends after $\bar{d}$ steps. Note also that, in $\R^{\bar{d}}$, $2\bar{d}$ elements are sufficient for \eqref{conv0} to hold. Hence we can take, after $\bar{d}$ steps, the $2\bar{d}^2$ elements $B_1,\ldots, B_{\bar{d}}$, from the above procedure, with the corresponding vector of coefficients $\lambda$ in $\R^{2\bar{d}^2}$ (which might have some $0$ components). \\

We are only left to construct the random vector $S_t$. Let now $s_j$ such that $a_j(\omega)=(X_{t-1}(\omega)-F_t(\omega,s_j),s_j-S_{t-1}(\omega))\in\R\times\R^{d}$. For any $j=1,\ldots 2\bar{d}$ on $\{s_j\in\mathbb{S}_t\}$ we may simply take $s_j$. If this is not possible $s_j$ is obtained as a limit of elements in $\mathbb{S}_t$. We can treat both cases simultaneously by defining, for any $n\in\mathbb{N}$, $X_j^n$ a measurable selector of $V_j^n:=s_j+\overline{B}_{\frac{1}{n}}\cap \mathbb{S}_t$ which is defined on $\{V_j^n\neq\varnothing\}$.\\

Note that there might exist $1\leq i,j\leq 2\bar{d}$, $n\in\mathbb{N}$ such that $X_i^n(\omega),X_j^n(\omega)\in\mathbb{S}_t(\omega)$ for the same $\omega\in\Omega$. By recalling that $\mathbb{S}_t$ is a convex set, we only need to replace $X_i^n,X_j^n$ with a suitable convex combination. Define
$$\widetilde{S}^n_t:=\sum_{j=1}^{n}\frac{\lambda_j}{\sum_{j=1}^{n}\lambda_j}  X_j^n\ ,$$ with $\lambda$ from the above procedure. Note finally that since $X^n_j$ is only defined on $\{V_j^n\neq\varnothing\}$ we need to take care of well-posedness when constructing $S_t$. Consider $\hat{S}_t$ 
an arbitrary measurable selector of $\mathbb{S}_t$ and set $S^1_t:=\widetilde{S}^1_t\mathbf{1}_{\cup_{j=1}^{\bar{d}} \{V_j^1\neq\varnothing\}}+\hat{S}_t\mathbf{1}_{(\cup_{j=1}^{\bar{d}} \{V_j^1\neq\varnothing\})^c}$. For $n> 1$ set $S^n_t:=\widetilde{S}^n_t\mathbf{1}_{\cup_{j=1}^{\bar{d}}\{ V_j^n\neq\varnothing\}}+S^{n-1}_t\mathbf{1}_{(\cup_{j=1}^{\bar{d}} \{V_j^n\neq\varnothing\})^c}$. The desired random vector is thus $S_t:=\lim_{n\rightarrow\infty}S^n_t$.
\end{proof}

\begin{proof}[Proof of Proposition \ref{propSprocess}]
Start with $S_0:=x_0$ and suppose first $F_0(x_0)<\infty$, which implies $F_t(\omega,\cdot)<\infty$ for any $t\in I$ and for any $\omega\in\Omega$. From \eqref{frictionleassDominance} if $F_0(x_0)=-\infty$ then the claim is trivial (see also Remark \ref{comments}). Suppose therefore $F_0(x_0)>-\infty$. Let $H_1$ be an optimal strategy for the $\mathbb{S}$-superhedging problem. From Lemma \ref{Sprocess} there exists $S_1$ such that $F_0(x_0)+H_1(\omega)\cdot\Delta S_1(\omega)\geq F_1(\omega,S_1(\omega))$ for every $\omega\in\Omega$. The random set $\mathcal{H}_2(\cdot,S_{1}(\cdot))$ is $\F_1$-measurable as it coincides with $\mathcal{H}^M_u$ from Corollary \ref{corSuperHmultif} in the Appendix with $u=2$, $X_{u-1}=S_{u-1}$, $X_u=\mathbb{S}_u$, $C=\R^d$. Let $H_2$ a measurable selector. Applying iteratively Lemma \ref{Sprocess} and Corollary \ref{corSuperHmultif} we get the inequalities
\begin{eqnarray*}
F_0(x_0)+H_{1}(\omega)\cdot \Delta S_{1}(\omega) &\geq & F_1(\omega,S_1(\omega)) \ ,\\
F_0(x_0)+H_{1}(\omega)\cdot \Delta S_{1}(\omega)+H_{2}\cdot \Delta S_{2}(\omega) &\geq & F_2(\omega,S_2(\omega))\ , \\
&\ldots & \\
F_0(x_0)+\sum_{t=1}^TH_{t}(\omega)\cdot \Delta S_{t}(\omega) &\geq & F_T(\omega,S_T(\omega))=g(\omega)\ , \\
\end{eqnarray*}
for some $S_1,\ldots, S_T$, $H_1,\ldots H_T$, and for every $\omega\in A$ with $A:=\{\omega\in \Omega\mid F_t(\omega,S_t(\omega))>-\infty\ \forall t=0,\ldots T\}$. Note that, by construction, $F_t(\omega,S_t(\omega))=-\infty$ for some $t=0,\ldots T$ if and only if $Q(\{\omega\})=0$ for every $Q\in\mathcal{Q}_S$, so that $A=\Omega_*(S)$.
 $F_0(x_0)$ is the cheapest super-hedge from the minimality of $F_{t}(\cdot,S_t(\cdot))$ for $t=0,\ldots T$. Obviously $S$ belongs to the bid-ask spread since $S_t\in\mathbb{S}_t$ for every $t$.\\

Suppose now that $F_0(x_0)=\infty$. Recall that, as in the proof of \ref{Fmeascont}, if $F_{s}(\omega,x)=\infty$ for some $s\in I$, $x\in\R^d$ then $F_{s}(\omega,\cdot)\equiv\infty$. Let $t:=\min\{s\in I\mid F_{s}(\omega,\cdot)<\infty\  \forall\omega\in\Omega\}\geq 1$.\\ Choose arbitrarily $S_u\in\mathcal{L}^0(\mathcal{F}_u;C_u)$ for $u=0,\ldots,t-1,t+1,\ldots T$, we need to define $S_t$.\\
Fix $\omega\in\Omega$ such that $F_{t-1}(\omega,\cdot)\equiv\infty$. For all $y\in\R$, consider the set
\begin{equation*}
\Gamma _{y}(\mathbb{S}_t):= co\left(conv\left\{ \left[ s-S_{t-1}(\widetilde{\omega})\ ;\ y-F_{t}(\widetilde{\omega },s)\right]\mid s\in \mathbb{S}_{t}(\widetilde{\omega}),\ \widetilde{\omega
}\in \Sigma _{t-1}^{\omega }\right\}\right)\subseteq \R^{d+1}.
\end{equation*}
Observe first that if for a finite set $\{\omega_1,\ldots \omega_k\}\subseteq\Sigma_{t-1}^{\omega}$ (or for the empty set) we have $0\notin int(\Gamma _{y}(\mathbb{S}_t\setminus U))$ with $U:=\{\mathbb{S}_t(\omega_1),\ldots \mathbb{S}_t(\omega_k)\}$
 then there exists $(H,h)\setminus(0,0)\in \mathbb{R}^{d}\times \mathbb{R}$ , such that
\begin{equation}\label{seprationSH}
h(y-F_{t}(\widetilde{\omega },s))+H\cdot (s-S_{t-1}(\widetilde{\omega}))\geq 0.
\end{equation}
If $h>0$ then $y+H/h\cdot (s-S_{t-1}(\widetilde{\omega}))\geq F_{t}(\widetilde{\omega },s)$ for all such $s$. From the continuity of $F_t(\omega,\cdot)$ (see Proposition \ref{Fmeascont}) and from $\mathbb{S}_{t}$ being closed and bounded we have that the quantities 
\begin{equation}\label{Noinf}
l_j:=\min\{y+H/h\cdot (s-S_{t-1}(\widetilde{\omega}))-F_{t}(\omega_j,s)\mid s\in\mathbb{S}_t(\omega_j)\}<0,\qquad l:=-\min_j l_j
\end{equation}
are well defined and finite. Observe now that $(y+l,H/h)$ solves the $\mathbb{S}$-superhedging problem of Definition \ref{defF} which is a contradiction since $F_{t-1}(\omega,S_{t-1}(\widetilde{\omega}))=\infty$.\\

Start with $y_1\in\R$. Since $\Gamma _{y_1}(\mathbb{S}_t)\subseteq \R^{d+1}$, there exist a finite number of vectors $U_1:=\{s_1,\ldots,s_{k_1}\}\subseteq \overline{\mathbb{S}_t(\Sigma_{t-1}^{\omega})}$ such that $\overline{\Gamma _{y_1}(U_1)}=\overline{\Gamma _{y_1}(\mathbb{S}_t)}$. In particular, from the above discussion, if \eqref{seprationSH} is satisfied for every $s\in U_1$ then $h\leq 0$.

For any $j=1,\ldots, k_1$, $s_j=\lim_{n\rightarrow\infty}s_j^n$ for some $s_j^n\in\mathbb{S}_t(\omega_j^n)$. If $s_j^n$ eventually belong to $\mathbb{S}_{t}(\omega_j)$ for some $\omega_j$, the sequence $s_j^n$ can be taken constantly equal to $s_j$ since $\mathbb{S}_{t}(\omega_j)$ is closed. Moreover, with no loss of generality, if $s_i,s_j\in \mathbb{S}_t(\Sigma_{t-1}^{\omega})$ we may suppose that the corresponding $\omega_i$, $\omega_j$ satisfy $\mathbb{S}_{t}(\omega_i)\neq\mathbb{S}_{t}(\omega_j)$ for $i\neq j$. Indeed, by the previous considerations, having $s_1,\ldots, s_l$ it is possible to find $s_{l+1}$ in $\mathbb{S}_t(\Sigma _{t-1}^{\omega })\setminus \{\mathbb{S}_t(\omega_1),\ldots \mathbb{S}_t(\omega_l)\}$ (see the discussion for \eqref{Noinf}). If $s_i\in  \overline{\mathbb{S}_t(\Sigma_{t-1}^{\omega})}\setminus \mathbb{S}_t(\Sigma_{t-1}^{\omega})$ we may suppose that $s_i^n\in \mathbb{S}_{t}(\omega_i^n)$ with $\omega_i^n\neq\omega_j^m$ for any $m\neq n$, $j\neq i$.\\ Let $E_1:=\cup_{j=1}^{k_1}\cup_{n=1}^{\infty}\{\omega_j^n\mid s_j^n\in \mathbb{S}_t(\omega_j^n)\}$ and set $$S_t^1(\widetilde{\omega}):=\begin{cases}s^n_j&\text{ if }\widetilde{\omega}\in\Sigma_t^{\omega_j^n}\\ \hat{S}_t(\omega)&\text{ otherwise }\end{cases}$$ with $\hat{S}_t\in\mathcal{L}^0(\mathcal{F}_t;C_t)$ arbitrary. $S_t^1$ has the same measurability of $\hat{S}_t$ since they coincide up to an union of countably many measurable sets. Note that by construction $\overline{\Gamma_{y_1}(S^1_t)}=\overline{\Gamma_{y_1}(\mathbb{S}_t)}$ and hence, as in the discussion for \eqref{seprationSH} and \eqref{Noinf}, it is not possible to separate $\{0\}$ ans $\Gamma_{y_1}(S^1_t)$ with $(H,h)$ such that $h>0$. We thus have,
\begin{equation}\label{condSH}
\inf\{x\in\R\mid x+H\cdot(S^1_{t}(\widetilde{\omega})-S_{t-1}(\widetilde{\omega}))\geq F_t(\widetilde{\omega},S_{t}(\widetilde{\omega}))\quad\forall \widetilde{\omega}\in\Sigma_{t-1}^\omega\}\geq y_1.
\end{equation}
Define now $y_n:=y_1+n$. For any $n\in\mathbb{N}$ we can apply the same procedure which yields a collection $\{S^{n}_t\}_{n\in\mathbb{N}}$ with the property that \eqref{condSH} is satisfied with $S^{n}_t$ and $y_n$. Moreover with no loss of generality we can choose $U_{n+1}\supseteq U_{n}$ and hence $E_{n+1}\supseteq E_n$ in order to have $S^{n+1}_t=S^{n}_t$ on $E_n$. We therefore have that $S_t:=\lim_{n\rightarrow\infty} S^n_t$ is well defined and $$
\inf\{x\in\R\mid x+H\cdot(S_{t}(\widetilde{\omega})-S_{t-1}(\widetilde{\omega}))\geq F_t(\widetilde{\omega},S_{t}(\widetilde{\omega}))\quad\forall \widetilde{\omega}\in\Sigma_{t-1}^\omega\}\geq \sup_n y_n=\infty.$$
Since $F_t(\cdot,S_t(\cdot))$ is the (conditional) cheapest amount for superhedging $g$ at time $T$ we have that the superhedging price of $g$ for the process $S$ is infinite.
 \end{proof}

\begin{proof}[Proof of Proposition \ref{S2process}] 

Note first that the function $G:\Omega\times\R^d\mapsto\R$ defined by
\begin{equation}
G(\omega,x):=X_{t-1}(\omega)+H_{t}(\omega)\cdot(x-S_{t-1}(\omega))-F_t(\omega,x)
\end{equation}
is a Carath\'eodory map and since $\mathbb{S}_t$ is a closed valued \Ftt-measurable set, the set
\begin{equation}\label{defMinimum}
Y_{t}(\omega):=\inf \left\{X_{t-1}(\omega)+H_t(\omega)\cdot(s-S_{t-1}(\omega))-F_t(\omega,s)\mid s\in\mathbb{S}_t(\omega)\right\}
\end{equation}
is \Ftt-measurable from Lemma \ref{measInf} and Lemma \ref{cara} in the Appendix. From Theorem \ref{implicit}, the set $E:=\{\omega\in\Omega\mid \exists x\in \mathbb{S}_t(\omega)\text{ with } G(\omega,x)=Y_t(\omega)\}$ is \Ftt-measurable and there exists a measurable function $m:E\mapsto\R^d$ such that
\begin{equation}\label{argmin}
m(\omega)\in \mathbb{S}_t(\omega),\qquad G(\omega,m(\omega))=Y_t(\omega) ,\qquad \forall\omega\in E= \{|Y_t|<\infty\}.
\end{equation}
Note that on $\{|Y_t|=\infty\}$ we have $|F_t(\omega,\cdot)|\equiv \infty$ and hence the random vectors $S_t$ and $H_{t+1}$ can be chosen arbitrarily. In particular they can be chosen to satisfy the desired properties. We may therefore suppose, without loss of generality, that $Y_t(\omega)$ is a minimum for every $\omega\in\Omega$.

We now show that there exists $H_{t+1}\in\Ltt$ such that, for any $\omega\in\Omega$, $H_{t+1}(\omega)\in \mathcal{H}_{t+1}(\omega, m(\omega))$ and
\begin{itemize}
\item if $H_t^i(\omega)< H^i_{t+1}(\omega)$ then $m^i(\omega)=\oS^i(\omega)$ ;
\item if $H_t^i(\omega)> H^i_{t+1}(\omega)$ then $m^i(\omega)=\uS^i(\omega)$ ;
\end{itemize}
and hence the desired random vector is $\widehat{S}_{t}:=m$. The desired strategy $H_{t+1}$ is obtained by taking any measurable selector of $\mathcal{H}_u^M$ given by Corollary \ref{corSuperHmultif} with $u=t+1$, $X_{u-1}=m$, $X_u=\mathbb{S}_{u}$ and $$C=\bigotimes_{i=1}^d\left\{(-\infty,H^i_{t}]\mathbf{1}_{\{m^i=\uS^i\}}+[H^i_{t},\infty)\mathbf{1}_{\{m^i=\oS^i\}}\cup \{H^i_t\}\right\}.$$

We are only left to show that such a set $\mathcal{H}_u^M$ is non-empty for every $\omega\in\Omega$.\\

Fix $\omega\in\Omega$. For simplicity of notations we omit the dependence on $\omega$ as no confusion arise here. In particular, $m=m(\omega)$ and $\mathcal{H}_{t+1}(x)=\mathcal{H}_{t+1}(\omega,x)$, $F_t(x)=F_t(\omega,x)$ for every $x\in\R^d$.\\

\bigskip

\textbf{Step 1.} 
Observe that for any $\widetilde{H}\in \mathcal{H}_{t+1}(m)$
\begin{equation}\label{minimizing_seq}
\inf\left\{F_{t}(m)+\widetilde{H}\cdot(s-m)- F_{t+1}(\widetilde{\omega},s)\mid s\in\mathbb{S}_{t+1}(\widetilde{\omega}),\ \widetilde{\omega}\in \Sigma_t^\omega \right\}=0
\end{equation}
otherwise $(F_t(m),\widetilde{H})$ would not be optimal. Let $\{y_n\}_{n=1}^{\infty}\subseteq \mathbb{S}_{t+1}(\Sigma_t^\omega)$ a minimizing sequence with corresponding $\{\widetilde{\omega}_n\}$ such that $y_n\in\mathbb{S}_{t+1}(\widetilde{\omega}_n)$. By denoting $y:=\lim_{n\rightarrow\infty} y_n$ and $f(y):=\lim_{n\rightarrow\infty}F_{t+1}(\widetilde{\omega}_n,y_n)$, we have that,
\begin{equation}\label{attained2}
F_t(m)+\widetilde{H}\cdot(y-m)=f(y).
\end{equation}

Let 
\begin{equation}\label{minimizersSet}
Y:=\left\{\lim_{n\rightarrow\infty} y_n \mid \{y_n\}_{n=1}^{\infty}\subseteq \mathbb{S}_{t+1}(\Sigma_t^\omega)\text{ and }\eqref{attained2} \text{ is satisfied }\right\}.
\end{equation}

In a first step we show that, for any $y\in conv(Y)$, $\widetilde{H}$ is still optimal for the (conditional) $\mathbb{S}$-superhedging problem with initial value $y$, that is, $\widetilde{H}\in\mathcal{H}_{t+1}(y)$.\\

Take $y:=\sum_{i=1}^n\lambda_i y_i\in conv(Y)$. The (conditional) $\mathbb{S}$-superhedging price $F_t(y)$ must satisfy, in particular, the constraints
$$x+\alpha\cdot(y_i-y)\geq f(y_i)\qquad \forall i=1,\ldots,n$$
and hence $F_t(y)\geq \sum_{i=1}^n\lambda_i f(y_i)$. Note however that $\widetilde{H}$ satisfies
\begin{eqnarray}
F_t(m)+\widetilde{H}\cdot(s-m)&\geq& F_{t+1}(\widetilde{\omega},s)\qquad \forall s\in\mathbb{S}_{t+1}(\widetilde{\omega}),\ \forall\widetilde{\omega}\in \Sigma_t^\omega\ ,\\
F_t(m)+\widetilde{H}(s-y)+\widetilde{H}\cdot(y-m)&\geq& F_{t+1}(\widetilde{\omega},s)\qquad \forall s\in\mathbb{S}_{t+1}(\widetilde{\omega}),\ \forall\widetilde{\omega}\in \Sigma_t^\omega\label{correction}\ ,\\
\sum_{i=1}^n\lambda_i f(y_i)+\widetilde{H}\cdot(s-y)&\geq& F_{t+1}(\widetilde{\omega},s)\qquad \forall s\in\mathbb{S}_{t+1}(\widetilde{\omega}),\ \forall\widetilde{\omega}\in \Sigma_t^\omega\ ,
\end{eqnarray}
where the last inequality follows from the fact that \eqref{attained2} holds for every $y_i$ with $i=1,\ldots,n$ and hence
$$F_t(m)+\widetilde{H}\cdot(y-m)=\sum_{i=1}^n\lambda_i\left(F_t(m)+\widetilde{H}\cdot(y_i-m)\right)=\sum_{i=1}^n\lambda_i f(y_i).$$ We have therefore that $\widetilde{H}\in\mathcal{H}_{t+1}(y)$.\\

\textbf{Step 2} We now prove that for any $y_0,y_1\in\R^d$, for any $0\leq\lambda\leq 1$
\begin{equation}
\mathcal{H}_{t+1}(y_0)\cap\mathcal{H}_{t+1}(y_1)\subseteq\mathcal{H}_{t+1}((1-\lambda) y_0+\lambda y_1 )
\end{equation}
and, moreover,
\begin{equation}\label{eqPrice}
F_t((1-\lambda) y_0+\lambda y_1)=F_t(y_0)+\lambda\widetilde{H}\cdot( y_1-y_0).
\end{equation}

Denote $y_{\lambda}:=(1-\lambda) y_0+\lambda y_1 $. Let $\widetilde{H}\in \mathcal{H}_{t+1}(y_0)\cap\mathcal{H}_{t+1}(y_1)$. We need to show that $\widetilde{H}$ is optimal for the (conditional) $\mathbb{S}$-superhedging problem with initial value $y_{\lambda}$.  For $\lambda=0,1$ the claim is trivial. Note that similarly as in \eqref{correction}, for any $0\leq\lambda\leq 1$, the following holds
\begin{eqnarray*}
F_t(y_0)+\widetilde{H}\cdot(y_{\lambda}-y_0)+\widetilde{H}(s-y_\lambda) &\geq& F_{t+1}(\widetilde{\omega},s)\qquad \forall s\in\mathbb{S}_{t+1}(\widetilde{\omega}),\ \forall\widetilde{\omega}\in \Sigma_t^\omega.
\end{eqnarray*}

Suppose that for some $\overline{\lambda}\in(0,1)$ this is not optimal and hence there exists a dominating strategy $H_{\bar{\lambda}}$ with
\begin{equation}\label{contraHP}
F_t(y_{\bar{\lambda}})<F_t(y_0)+\widetilde{H}\cdot(y_{{\bar{\lambda}}}-y_0).
\end{equation}
 From
 \begin{eqnarray*}
  F_t(y_{\bar{\lambda}})+H_{\bar{\lambda}}(y_0-y_{{\bar{\lambda}}})+H_{\bar{\lambda}}(s-y_0)&\geq& F_{t+1}(\widetilde{\omega},s)\qquad \forall s\in\mathbb{S}_{t+1}(\widetilde{\omega}),\ \forall\widetilde{\omega}\in \Sigma_t^\omega\ ,\\
    F_t(y_{\bar{\lambda}})+H_{\bar{\lambda}}(y_1-y_{{\bar{\lambda}}})+H_{\bar{\lambda}}
    (s-y_1)& \geq& F_{t+1}(\widetilde{\omega},s)\qquad \forall s\in\mathbb{S}_{t+1}(\widetilde{\omega}),\ \forall\widetilde{\omega}\in \Sigma_t^\omega \ , 
  \end{eqnarray*}
we get
\begin{eqnarray}
  F_t(y_0)&\leq & F_t(y_{\bar{\lambda}})+H_{\bar{\lambda}}(y_0-y_{{\bar{\lambda}}})\label{supIn1}\ ,\\
    F_t(y_1)=F_t(y_0)+\widetilde{H}(y_1-y_0)& \leq & F_t(y_{\bar{\lambda}})+H_{\bar{\lambda}}(y_1-y_{{\bar{\lambda}}})\ .\label{supIn2}
 \end{eqnarray}

 From \eqref{contraHP} and \eqref{supIn1} we have $(\widetilde{H}-H_{\bar{\lambda}})(y_{{\bar{\lambda}}}-y_0)>0$. As $y_{{\bar{\lambda}}}-y_0=\lambda(y_1-y_0)$ we thus obtain
 \begin{equation}\label{contra1} (\widetilde{H}-H_{\bar{\lambda}})(y_1-y_0)>0.
 \end{equation}
 Now, from \eqref{contraHP} and \eqref{supIn2} we get

 $\widetilde{H}(y_1-y_0)<\widetilde{H}(y_{{\bar{\lambda}}}-y_0)+H_{\bar{\lambda}}(y_1-y_{{\bar{\lambda}}})$ from which $\widetilde{H}(y_1-y_{{\bar{\lambda}}})<H_{\bar{\lambda}}(y_1-y_{{\bar{\lambda}}})$. Since $y_1-y_{{\bar{\lambda}}}=(1-\lambda)(y_1-y_0)$ we thus obtain
  \begin{equation}\label{contra2} (\widetilde{H}-H_{\bar{\lambda}})(y_1-y_0)<0.
 \end{equation}
Equation \eqref{contra2} clearly contradicts \eqref{contra1}.\\
The assertion in \eqref{eqPrice} follows from the contradiction of \eqref{contraHP}.\\

\textbf{Step 3} We now conclude the proof of the Proposition.
As $H\in\mathcal{H}_t(\omega)$ is fixed, for simplicity, we can translate $H$ in the origin. Denote by
\begin{eqnarray*}
I_u&:=&\{i\in\{1,\ldots d\}\mid m^i=\oS^i(\omega)\}\\
I_d&:=&\{i\in\{1,\ldots d\}\mid m^i=\uS^i(\omega)\}\\
\xi_i&:=&\textbf{1}_{I_u}(i)-\textbf{1}_{I_d}(i)
\end{eqnarray*} and define $$R:=\xi_1[0,\infty)\times,\ldots\times\xi_d[0,\infty)$$
where with a slight abuse of notation $\xi_i[0,\infty)$ is either $[0,\infty)$, $(-\infty,0]$ or $\{0\}$ according to $\xi_i$ being respectively $1$,$-1$ or $0$.\\
Suppose that there is no $\widetilde{H}\in \mathcal{H}_{t+1}(m)$ that meets the requirement, that is  $$\mathcal{H}_{t+1}(m)\cap R=\varnothing.$$ As $\mathcal{H}_{t+1}(m)$ and $R$ are both closed convex sets in $\R^d$, by Hahn Banach Theorem, there exists $\eta\in\R^d$, $\gamma\in\R$ such that $$\eta\cdot \widetilde{H}\geq \gamma> \sup_{r\in R} \eta\cdot r  \qquad\forall \widetilde{H}\in \mathcal{H}_{t+1}(m).$$
Note that $\forall i\in I_u$ and $\forall\alpha\geq 0$ we have that $\alpha e_i\in R$ where $e_i$ is the $i^{th}$ element of the canonical basis of $\R^d$. Since $\sup_{r\in R} \eta\cdot r$ is bounded from above we infer that $\eta_i\leq 0$ if $i\in I_u$. Similarly  $\eta_i\geq 0$ if $i\in I_d$. Any separator $\eta$ must therefore satisfy
\begin{eqnarray}
\eta_i\leq 0 & \text{ if } & i\in I_u\label{bid_ask_direction1}\ ,\\
\eta_i\geq 0 & \text{ if } & i\in I_d\label{bid_ask_direction2}\ .
\end{eqnarray}

 Note moreover that as $0\in R$
\begin{equation}\label{positivity}
 \eta\cdot \widetilde{H}> 0 \qquad\forall \widetilde{H}\in \mathcal{H}_{t+1}(m).
\end{equation}

Denote by $l:=d(\mathcal{H}_{t+1}(m),R)$ the distance between the two sets and denote by $\widehat{H},\hat{r}$ the minimizers which exist since $\mathcal{H}_{t+1}(m)$ and $R$ are closed.
Let $Y=Y(\widehat{H})$ as in \eqref{minimizersSet} in Step 1 and introduce the convex cone $V:=co\left(conv\{y-m\mid y\in Y\}\right)$. 

Note that by definition of $Y$ in \eqref{minimizersSet}, any $(1,y-m)$ with $y\in Y$ defines a supporting hyperplane for the convex set of acceptable couples $\mathcal{A}_{t+1}(\omega,m)\subseteq \R^{d+1}$ (see Definition \ref{defF}) at $(F_t(m),\widehat{H})$. In particular, any $y-m$ with $y\in Y\setminus\{m\}$ defines a supporting hyperplane for $\mathcal{H}_{t+1}(m)\subseteq \R^d$ at $\widehat{H}$. If $Y\setminus\{m\}= \varnothing$ then $\mathcal{H}_{t+1}=\R^d$ and we already have a contradiction. If $Y\setminus\{m\}\neq \varnothing$ we have that $$co(\widetilde{H}-\widehat{H}\mid \widetilde{H}\in\mathcal{H}_{t+1}(m))=V^*.$$
Observe now that $\eta\in V^{**}=V$ and hence $\eta=\alpha(y-m)$, for some $y\in conv(Y)$, $\alpha>0$. Since $\frac{1}{\alpha}\eta\in V$, with no loss of generality assume $\alpha=1$.

Equations \eqref{bid_ask_direction1} and \eqref{bid_ask_direction2} imply that
\begin{eqnarray}
y^i_t\leq m^i & \text{ if }  i\text{ is such that }& m^i=\oS^i(\omega)\label{bid_ask_direction3},\\
y^i_t\geq m^i & \text{ if }  i\text{ is such that }& m^i=\uS^i(\omega)\label{bid_ask_direction4}.
\end{eqnarray}
Since $\widehat{H}\in\mathcal{H}_{t+1}(m)$, from Step 1, we have $\widehat{H}\in\mathcal{H}_{t+1}(y)$. Thus, from Step 2, $\widehat{H}\in\mathcal{H}_{t+1}(\lambda m+(1-\lambda)y)$ is also true for every $0\leq\lambda\leq 1$. From \eqref{bid_ask_direction3} and \eqref{bid_ask_direction4} there exists $\lambda$ sufficiently close to $1$ such that $y_\lambda:=(1-\lambda) m+\lambda y\in C_t$ and, from \eqref{eqPrice} in Step 2,
\begin{equation}\label{translatedY}
F_t(y_\lambda)= F_t(m)+\widehat{H}(y_\lambda-y_0).
\end{equation}
Note moreover that, by construction, $y,m\in \overline{conv}(\mathbb{S}_{t+1}(\Sigma_t^{\omega}))$ and hence $y_\lambda\in \overline{conv}(\mathbb{S}_{t+1}(\Sigma_t^{\omega}))\cap C_t=\mathbb{S}_{t}$. By translating back $0$ in $H$, equation \eqref{positivity} implies that $\widehat{H}\cdot(y_\lambda-m)>H\cdot(y_\lambda-m)$. In combination with \eqref{translatedY} and the fact that $F_t(m)=X_{t-1}-Y_t+H\cdot(m-S_{t-1})$ from equations \eqref{defMinimum} and \eqref{argmin}, we thus obtain

\begin{eqnarray*}
F_t(y_\lambda)&=&F_t(m)+ \widehat{H}\cdot(y_\lambda-m)\\
&>&F_t(m)+ H\cdot(y_\lambda-m)\\
&=&X_{t-1}-Y_t+H\cdot(m-S_{t-1})+H\cdot(y_\lambda-m)\\
&=&X_{t-1}-Y_t+H\cdot(y_\lambda-S_{t-1})\ ,
\end{eqnarray*}
which is a contradiction since $y_\lambda\in \mathbb{S}_t$ and $Y_t$ is a minimum in \eqref{defMinimum}.
\end{proof}

\begin{proof}[Proof of Corollary \ref{corS2process}]
Note first that if $F_0(x_0)<\infty$ then $F_t(\omega,\cdot)<\infty$ for any $t\in I$. Applying iteratively Proposition \ref{S2process}, there exists a process $\hat{S}$ with $\hat{S}_0=x_0$ and a strategy $H$ which satisfy the following inequalities
\begin{eqnarray*}
F_0(x_0)+H_{1}(\omega)\cdot \Delta \hat{S}_{1}(\omega) &\geq & F_1(\omega,\hat{S}_1(\omega)) \\
F_0(x_0)+H_{1}(\omega)\cdot \Delta \hat{S}_{1}(\omega)+H_{2}\cdot \Delta \hat{S}_{2}(\omega) &\geq & F_2(\omega,\hat{S}_2(\omega)) \\
&\ldots & \\
F_0(x_0)+\sum_{t=1}^TH_{t}(\omega)\cdot \Delta \hat{S}_{t}(\omega) &\geq & F_T(\omega,\hat{S}_T(\omega))=g(\omega) \\
\end{eqnarray*}
on $A:=\cap_{t=0}^{T} \{\omega\in \Omega\mid \exists s\in\mathbb{S}_t(\omega)\text{ such that }F_t(\omega,s)>-\infty\}$. Note that, by construction, $F_t(\omega,s)=-\infty$ for every $s\in\mathbb{S}_t(\omega)$ if and only if $Q(\{\omega\})=0$ for every $Q\in\mathcal{Q}$, so that $A=\Omega_*$ (see also Remark \ref{comments}).
Rearranging the terms in the summation as $$\sum_{t=1}^TH_{t}\cdot \Delta \hat{S}_{t}=\sum_{t=1}^T(H_{t}-H_{t+1})\cdot\hat{S}_{t}-H_1\cdot x_0$$ the properties of $\hat{S}$ yield the desired inequality. 
\end{proof}

 \section{Appendix}
Let $(\Omega,\mathcal{A})$ a measurable space.
\begin{definition}\label{defRandom}
 A map $\Psi:\Omega\mapsto 2^{\R^n}$, where $2^{\R^n}$ is the power set of $\R^n$,  is called multi-function, or random set. It is said to be $\mathcal{A}$-measurable if, for any open $O\subseteq \R^n$ the set $\{\omega\in\Omega\mid \Psi(\omega)\cap O\neq\varnothing\}$ is $\mathcal{A}$-measurable.
\end{definition}

\begin{lemma}\label{meas-eps} Let $\Psi:\Omega\mapsto 2^{\R^n}$ a $\mathcal{A}$-measurable multi-function. Let $\varepsilon>0$ then $$\Psi^\varepsilon:\omega\mapsto \left\{v\in \R^n\mid v\cdot s\geq \varepsilon \quad \forall s\in\Psi(\omega)\setminus \{0\}\right\}$$ is an $\mathcal{A}$-measurable multi-function.
\end{lemma}
\begin{proof} see Appendix of \cite{BFM14}
\end{proof}

\begin{theorem}\label{castaing}[Theorem 14.5 \cite{R}] The following are equivalent
\begin{itemize}
\item $\Psi:\Omega\mapsto 2^{\R^n}$ is a closed valued, $\mathcal{A}$-measurable multi-function
\item $\Psi$ admits a Castaing representation: there is a countable family $\{\psi_n\}_{n\in\mathbb{N}}$ of $\mathcal{A}$-measurable function $\psi_n:\textrm{dom} \Psi\mapsto\R^n$ such that for any $\omega\in\Omega$ $$\Psi(\omega)=\textrm{cl }\{\psi_n(\omega)\mid n\in\mathbb{N}\}.$$
\end{itemize}
\end{theorem}
\begin{proposition}\label{preservation_measurability}[Propositions 14.2-14.11-14.12 \cite{R}] Consider a class of $%
\mathcal{A}$-measurable multi-functions. The following
operations preserve $\mathcal{A}$-measurability: countable unions,
countable intersections (if the functions are closed-valued),
finite linear combination, convex/linear/affine hull, generated
cone, polar set, closure, cartesian product of a finite number of $\mathcal{A}$-measurable multi-functions.
\end{proposition}

\begin{lemma}\label{measInf} Let $A$ be a real-valued, $\mathcal{A}$-measurable random set. Then $\inf A$ is $\mathcal{A}$-measurable.
\end{lemma}
\begin{proof}
For any $y\in\mathbb{R}$
\begin{equation*}
\left\{ \omega \in \Omega \mid \inf \{a\mid a\in A(\omega )\}<y\right\}
=\left\{ \omega \in \Omega \mid A(\omega )\cap (-\infty ,y)\neq \varnothing
\right\}\in\mathcal{A}
\end{equation*}
from which the thesis follows.
\end{proof}

\begin{theorem}\label{closed_value_selection}[Corollary 14.6 \cite{R}] A closed-valued measurable mapping always admits a measurable selector.
\end{theorem}

\begin{lemma}\label{cara}[Example 14.15 in \cite{R}] Let $F:\Omega\times\R^n\mapsto\R^m$ be a Carath\'eodory map and let $X(\omega)\subseteq \R^n$ be closed-valued and $\mathcal{A}$-measurable then the following maps are $\mathcal{A}$-measurable
\begin{itemize}
\item $\omega\mapsto F(\omega,X(\omega))$
\item $\omega\mapsto (X(\omega),F(\omega,X(\omega)))$
\end{itemize}
\end{lemma}

\begin{theorem}\label{implicit}[Theorem 14.16 in \cite{R}] Let $F:\Omega\times\R^n\mapsto\R^m$ be a Carath\'eodory map and let $X(\omega)\subseteq \R^n$ and $D(\omega)\subseteq \R^m$ be closed sets that depends measurably on $\omega$. Then the set
\begin{equation*}
E:=\left\{\omega\in\Omega\mid \exists x\in X(\omega)\text{ with } F(\omega,x)\in D(\omega)\right\}
\end{equation*}
is measurable and there exists a measurable function $x:E\mapsto\R^n$ such that
$$x(\omega)\in X(\omega)\quad\text{and}\quad F(\omega,x(\omega))\in D(\omega)\ \forall\omega\in E.$$
\end{theorem}

\begin{lemma}\label{compl} Let $A$ be a $\mathcal{A}$-measurable, closed-valued, random set. $A^C$ is $\mathcal{A}$-measurable.
\end{lemma}
\begin{proof} Since $A$ is closed valued and measurable, $d:\Omega\times \R^n\mapsto \R$ given by $d(\omega,x):=dist(x,A(\omega))$ is a Carath\'eodory map. From Lemma \ref{cara}, for every $O\subseteq\R^n$ open, $B(\omega):=d(O,A(\omega))$ is a $\mathcal{A}$-measurable subset of $\R$. From Lemma \ref{measInf} the function $\sup B$ is also $\mathcal{A}$-measurable, therefore,
$$\{\omega\in\Omega\mid A^C(\omega)\cap O\neq\varnothing\}=\{\omega\in\Omega\mid \sup B(\omega)>0\}\in\mathcal{A}.$$
\end{proof}

\begin{lemma}\label{measSuperHmultif} Let $1\leq u\leq T$. Let $\varphi:\Omega\rightrightarrows[-\infty,+\infty]$ and $\Delta_u:\Omega\rightrightarrows \R^n$ be multi-functions measurable with respect to \F, $\F_u$, respectively. Given a closed valued, $\F_{u-1}$-measurable random set of constraints $D\subseteq \R^{n}$, the following multi-function is $\F_{u-1}$-measurable
\begin{equation*}
A_D(\omega )=\left\{  (H,y)\in D\times \R \mid y+H\cdot\delta\geq f\qquad \forall \widetilde{\omega
}\in \Sigma _{u-1}^{\omega },\ \forall (\delta,f)\in\Delta_{u}(\widetilde{\omega })\times\varphi(\widetilde{\omega })\right\}.
\end{equation*}
Moreover, denoting with $\Pi _{x_{1},\ldots
,x_{n}}(\cdot)$ and $\Pi _{x_{n+1}}(\cdot)$ the canonical projection on the first $n$ components and on the $(n+1)^{th}$ component, respectively, we have that
\begin{equation*} M_{u-1}=\min\Pi _{x_{n+1}}(A_D),\qquad \mathcal{H}^M_u=\Pi _{x_{1},\ldots
,x_{n}}\left(A_D\cap \left\{ \mathbb{R}^{n}\times
M_{u-1}\right\} \right),
\end{equation*}
are also $\F_{u-1}$-measurable and closed valued. In addition, $M_{u-1}$ is single-valued.

\end{lemma}
\begin{proof}
We first show the measurability of the following multi-function
\begin{equation*}
\psi:\omega \mapsto \left\{ \Delta_u(\widetilde{\omega }%
)\times \{1\}\times \varphi(\widetilde{\omega })\mid \widetilde{\omega }\in
\Sigma _{u-1}^{\omega }\right\} \subseteq \mathbb{R}^{n+2}.
\end{equation*}
Let $O\subseteq \mathbb{R}^{n}\times \mathbb{R}^{2}$ be an open set and define $B:=\left\{\omega\in\Omega\mid \{\Delta_{u}(\omega)\times \{1\}\times \varphi(\omega)\}\cap O\neq \varnothing\right\}\in \F$. Note now that if $\omega$ satisfies $\psi (\omega)\cap O\neq \varnothing$, any $\widetilde{\omega}\in\Sigma^{\omega}_{u-1}$ satisfies the same. Define the function $\gamma_{u-1}:=(\uS_{0:u-1},\oS_{0:u-1})$ and recall that $S_{0:u-1}(\omega)$ is a shorthand for the trajectory of the process $S$ up to time $u-1$. The set $\gamma_{u-1}^{-1}(\gamma_{u-1}(B))\subset \Omega$ contains those $\Sigma^{\omega}_{u-1}$ for which there exists $\omega_b\in B$ with $\uS_{0:u-1}(\omega)=\uS_{0:u-1}(\omega_b)$ and $\oS_{0:u-1}(\omega)=\oS_{0:u-1}(\omega_b)$.  We thus have
\begin{equation*}
\{\omega \in \Omega \mid \psi (\omega )\cap O\neq \varnothing
\}=\gamma_{u-1}^{-1}(\gamma_{u-1}(B))\in\F_{u-1},
\end{equation*}
from which $\psi$ is $\F_{u-1}$-measurable.\\
By preservation of measurability (again Proposition \ref{preservation_measurability}) the
multi-function
\begin{equation*}
\psi^{\ast }(\omega ):=\left\{ H\in \mathbb{R}^{n+2}\mid H\cdot
y\geq 0\quad \forall y\in \psi(\omega )\right\}
\end{equation*}%
is also $\F_{u-1}$-measurable and thus, the same holds for $%
\psi^{\ast }\cap D\times \R\times \{-1\}$. It is easy to see now that $A_D=\Pi _{x_{1},\ldots ,x_{n+1}}(\psi
^{\ast }\cap D\times \R\times \{-1\})$ which is measurable from the continuity of the projection maps.\\
Observe now that the measurability of $A_D$ implies now those of $M_{u-1}$ and $\mathcal{H}^M_u$. Indeed $A:=\Pi _{x_{n+1}}(A_D)$ is again measurable by the continuity of projections. By taking the infimum of the real random set $A$ the measurability is preserved from Lemma \ref{measInf}. As in the classical case, the infimum, when finite, is actually a minimum by repeating (for example) the same arguments as in Proposition 2.1 in \cite{BFM15}. Finally $\mathcal{H}^M_u$ is again $\F_{u-1}$-measurable by preservation of measurability.
\end{proof}
\begin{corollary}\label{corSuperHmultif} Let $1\leq u\leq T$. Let $X_{u-1}$ be an $\F_{u-1}$-measurable function and $X_u:\Omega\rightrightarrows \R^n$ an $\F_u$-measurable multi-function. Suppose that $F_u(\cdot,\cdot)$ is a Carath\'eodory map on $D_{F_u}\subseteq\Omega$ and $X_u$ takes values in $D_{F_u}$. Given a closed valued, $\F_{u-1}$-measurable, random set of constraints $C\subseteq \R^{n}$, the following multi-function is $\F_{u-1}$- measurable
\begin{equation*}
A_C(\omega )=\left\{  (H,y)\in C\times \R \mid y+H\cdot(x_u-X_{u-1}(\omega))\geq F_{u}(\widetilde{\omega },x_u)\quad \forall x_u\in X_u(\widetilde{\omega}),\ \widetilde{\omega
}\in \Sigma _{u-1}^{\omega }\right\}.
\end{equation*}%
Moreover, denoting with $\Pi _{x_{1},\ldots
,x_{n}}(\cdot)$ and $\Pi _{x_{n+1}}(\cdot)$ the canonical projection on the first $n$ components and on the $(n+1)^{th}$ component, respectively, we have that
\begin{equation*} M_{u-1}=\min\Pi _{x_{n+1}}(A_C),\qquad \mathcal{H}^M_u=\Pi _{x_{1},\ldots
,x_{n}}\left(A_C\cap \left\{ \mathbb{R}^{n}\times
M_{u-1}\right\} \right)
\end{equation*}
are also $\F_{u-1}$-measurable and closed valued. In addition, $M_{u-1}$ is single-valued
\end{corollary}
\begin{proof}

Since $X_u$ takes value in $D_{F_{u}}$, Lemma \ref{cara} and Proposition \ref{preservation_measurability} imply that the multi-function $$\psi(\omega):\omega\mapsto\left(X_{u}(\omega),F_{u}(\omega,X_{u}(\omega)\right)-(X_{u-1}(\omega),0)\subseteq \R^{n+1}$$ is $\F_u$-measurable.
Apply now Lemma \ref{measSuperHmultif} with $\Delta_u=\psi$, $\varphi=0$ and $D=C\times \{-1\}$.
\end{proof}

\end{document}